\documentclass[letter,11pt]{article}
\usepackage{pdfpages}
\usepackage{fullpage}
\usepackage{textcomp}
\usepackage{setspace}
\usepackage{subcaption}
\usepackage{enumerate}
\usepackage{amsmath}
\usepackage{amsfonts}
\usepackage{graphicx}
\usepackage{tcolorbox}
\usepackage{amssymb}
\usepackage{multirow}
\usepackage{soul}
\usepackage{colortbl}
\usepackage{wrapfig}
\usepackage{algorithm}
\usepackage{algorithmicx}
\usepackage{algpseudocode}
\usepackage{amsthm}
\usepackage{todonotes}
\usepackage{mathtools}
\usepackage{mathrsfs}

\usepackage{thmtools} 
\usepackage{thm-restate}
\usepackage{physics}
\usepackage{bm}

\algtext*{EndWhile}% Remove "end while" text
\algtext*{EndIf}% Remove "end if" text
\algtext*{EndFor}% Remove "end for" text

\newtheorem{lemma}{Lemma}

\numberwithin{lemma}{section}
\newtheorem{theorem}[lemma]{Theorem}
\newtheorem{theorem*}[lemma]{Theorem*}

\newtheorem{claim}[lemma]{Claim}
\newtheorem{proposition}[lemma]{Proposition}

\title{Parameterized Approximation of Rectangle Stabbing}

\author{Huairui Chu\thanks{Department of Computer Science, University of California Santa Barbara. Email: huairuichu@ucsb.edu, aj66@ucsb.edu, daniello@ucsb.edu, anikait@ucsb.edu, tschibler@ucsb.edu, xiaoyangxu@ucsb.edu}\and Ajaykrishnan E S$^*$\and Daniel Lokshtanov$^*$\and Anikait Mundhra$^*$\and Thomas Schibler$^*$\and Xiaoyang Xu$^*$\and Jie Xue\thanks{Department of Computer Science, New York University Shanghai. Email: jiexue@nyu.edu}}

\date{}

\begin{document}

% \pagenumbering{gobble}

\maketitle

\bibliographystyle{plainurl}

\begin{abstract}
In the {\sc Rectangle Stabbing} problem, input is a set ${\cal R}$ of axis-parallel rectangles and a set ${\cal L}$ of axis parallel lines in the plane. 
The task is to find a minimum size set ${\cal L}^* \subseteq {\cal L}$ such that for every rectangle $R \in {\cal R}$ there is a line $\ell \in {\cal L}^*$ such that $\ell$ intersects $R$.
Gaur et al. [Journal of Algorithms, 2002] gave a polynomial time $2$-approximation algorithm, while Dom et al. [WALCOM 2009] and Giannopolous et al. [EuroCG 2009] independently showed that, assuming FPT $\neq$ W[1], there is no algorithm with running time $f(k)(|{\cal L}||{\cal R}|)^{O(1)}$ that determines whether there exists an optimal solution with at most $k$ lines. 
We give the first parameterized approximation algorithm for the problem with a ratio better than $2$. In particular we give an algorithm that given ${\cal R}$, ${\cal L}$, and an integer $k$ runs in time $k^{O(k)}(|{\cal L}||{\cal R}|)^{O(1)}$ and either correctly concludes that there does not exist a solution with at most $k$ lines, or produces a solution with at most $\frac{7k}{4}$ lines. 
We complement our algorithm by showing that unless FPT $=$ W[1], the {\sc Rectangle Stabbing} problem does not admit a $(\frac{5}{4}-\epsilon)$-approximation algorithm running in $f(k)(|{\cal L}||{\cal R}|)^{O(1)}$ time for any function $f$ and $\epsilon > 0$.
% \textbf{\color{red} Optional:} We also consider a natural generalization to 3 dimensions, and show a parameterized approximation ratio of 8/3.
\end{abstract}

% \newpage
% \pagenumbering{arabic}

\section{Introduction}
In this paper, we consider the {\sc Rectangle Stabbing} problem, which is a well-studied geometric covering problem defined as follows.
\begin{tcolorbox}[colback=gray!5!white,colframe=gray!75!black]
        \textsc{Rectangle Stabbing} \hfill \textbf{Parameter:} $k = |\mathcal{L}^*|$
        \vspace{0.1cm} \\
        \textbf{Input:} A set $\mathcal{R}$ of axis-parallel rectangles and a set $\mathcal{L}$ of axis-parallel lines in $\mathbb{R}^2$ \\
        \textbf{Goal:} Compute a minimum subset $\mathcal{L}^* \subseteq \mathcal{L}$ such that $\mathcal{R}$ is stabbed by $\mathcal{L}^*$
\end{tcolorbox}

Here a line $\ell \in {\cal L}$ {\em stabs} a rectangle $R \in {\cal R}$ if $\ell$ and $R$ have non-empty intersection. 
The problem was first considered by Hassin and Megiddo~\cite{hassin1991approximation}, who showed that the problem is NP-hard even when all the rectangles are horizontal line segments of unit length. 
In addition to being a very natural special case of the {\sc Set Cover} problem, the {\sc Rectangle Stabbing} problem generalizes the {\sc Optimal Discretization} problem where input is a set of points in $\mathbb{R}^2$, each point colored with one of multiple colors, and the task is to divide the plane using the minimum number of axis-parallel lines such that no cell contains points of different colors~\cite{cualinescu2005separating, kratsch2021optimal}. This has applications to fault tolerant sensor networks~\cite{cualinescu2005separating} and is also an idealized version of discretization problems considered in machine learning~\cite{garcia2012survey,kratsch2021optimal,witten2005practical}.
Tardos~\cite{tardos1995transversals} proved that for every set ${\cal R}$ of axis-parallel rectangles and integer $k$, either there exists a subset ${\cal R}^*$ of  ${\cal R}$ of size $k$ such that no line in ${\cal L}$ stabs two rectangles in  ${\cal R}^*$, or all the rectangles in  ${\cal R}$ can be stabbed using at most $2k$ lines. However the proof of Tardos is non-constructive in the sense that it does not lead to a polynomial time algorithm that produces one of the two outcomes. 
In 2002, Gaur et al.~\cite{gaur2002constant} gave a polynomial time factor $2$ approximation algorithm for the {\sc Rectangle Stabbing} based on an elegant rounding scheme for the relaxation of the most natural integer linear programming formulation of the problem.
Subsequently Ben-David et al.~\cite{ben2012approximability} showed an integrality gap of $2-\epsilon$ for every $\epsilon>0$ for the considered integer linear programming formulation. 
Despite significant attention to related variants of the problem~\cite{cualinescu2005separating,ChanD0SW18,eisenbrand2021qptas,even2008algorithms,khan2024approximation,khan2024ptas,kovaleva2006approximation,xu2007constant} from the perspective of approximation algorithms, the factor $2$ approximation of  Gaur et al.~\cite{gaur2002constant} remains the state of the art. 

The {\sc Rectangle Stabbing} problem, and variants thereof, have also been considered from the perspective of parameterized algorithms,
parameterized by the optimal value $k$. Dom et al.~\cite{dom2009parameterized} (see also~\cite{dom2008parameterized}) and Giannopolous et al.~\cite{giannopoulos2013fixed} independently showed that, assuming FPT $\neq$ W[1], there is no algorithm with running time $f(k)n^{O(1)}$ that determines whether there exists an optimal solution with at most $k$ lines. Here $n = \abs{\cal L} + \abs{\cal R}$, and the hardness results of~\cite{dom2009parameterized,giannopoulos2013fixed} apply even in the case when all the rectangles are squares of the same size. Giannopolous et al.~\cite{giannopoulos2013fixed} additionally show that the hardness result applies when ${\cal R}$ is a set of disjoint, but not axis-parallel rectangles. 
Both Dom et al.~\cite{dom2009parameterized} and Giannopolous et al.~\cite{giannopoulos2013fixed} give $k^{O(k)}n^{O(1)}$ time algorithms for the case when ${\cal R}$ is a set of {\em disjoint} axis-parallel squares of the same size, and pose as an open problem whether the case of disjoint axis-parallel rectangles admits an algorithm with running time $f(k)n^{O(1)}$. This was resolved in the affirmative by Heggernes et al.~\cite{heggernes2013fixed}, who gave an algorithm with running time $k^{O(k^2 \log k)}n^{O(1)}$ for this case.

The instances of {\sc Rectangle Stabbing} that are equivalent to instances of {\sc Optimal Discretization} have axis-parallel rectangles that may not be disjoint. Thus the algorithm of Heggernes et al.~\cite{heggernes2013fixed} does not imply an algorithm with running time $f(k)n^{O(1)}$ for the {\sc Optimal Discretization} problem, where $k$ is the solution size. An algorithm for {\sc Optimal Discretization} with running time $k^{O(k^2 \log k)}n^{O(1)}$ was given by Kratsch et al.~\cite{kratsch2021optimal}.

In this paper, we initiate the study of parameterized approximation algorithms (see the surveys~\cite{feldmann2020survey,kortsarz2016fixed,marx2008parameterized} for an introduction to parameterized approximation) for the {\sc Rectangle Stabbing} problem. Our main result is an algorithm with running time dependence on the solution size $k$ which is better than what is possible (assuming FPT $\neq$ W[1]) for parameterized algorithms computing exact solutions, while simultaneously achieving an approximation ratio that is substantially better than that of the best known polynomial time approximation algorithm~\cite{gaur2002constant}. Specifically we show the following. 
\begin{theorem} \label{thm-algorithm}
    There exists a $\frac{7}{4}$-approximation algorithm for \textsc{Rectangle Stabbing} with running time $k^{O(k)} \cdot n^{O(1)}$, where $n = |\mathcal{R}| + |\mathcal{L}|$.
\end{theorem}

The algorithm of Theorem~\ref{thm-algorithm} works by first guessing some features of the optimal solution, including the number of horizontal and vertical lines. 
These features are used to greedily select at most $k$ lines into the approximate solution, such that {\em (i)} the remaining rectangles not yet stabbed by the selected lines can be stabbed with only $\frac{3k}{4}$ additional lines, and {\em (ii)} the remaining instance can be solved in polynomial time by reduction to $2$-SAT. 
%\todo{please check this for bullshit! Seems fine to me. -Thomas}
%
We complement our main result by showing that the ratio of $\frac{7}{4}$ of Theorem~\ref{thm-algorithm} cannot be improved to a ratio arbitrarily close to one. 
\begin{theorem} \label{thm-hardness}
    %Assuming FPT $\neq$ W[1], for any constant $\epsilon > 0$ and any function $f$, there does not exist a $(\frac{5}{4}-\epsilon)$-approximation algorithm for \textsc{Rectangle Stabbing} which takes as input an instance $({\cal R}, {\cal L})$, an integer $k$,  runs in time $f(k) \cdot n^{O(1)}$, where $n = |\mathcal{R}| + |\mathcal{L}|$, and either concludes that no subset of at most $k$ lines of ${\cal L}$ stabs ${\cal R}$ or finds a set ${\cal L}^* \subseteq {\cal L}$ of size at most $\left(\frac{5}{4}-\epsilon\right)k$ that stabs ${\cal R}$.
    Assuming that $\textnormal{FPT}\neq\textnormal{W[1]}$, for any constant $\epsilon > 0$ and any function $f$, there does not exist a $(\frac{5}{4}-\epsilon)$-approximation algorithm for \textsc{Rectangle Stabbing} with running time $f(k) \cdot n^{O(1)}$, where $n = |\mathcal{R}| + |\mathcal{L}|$.
\end{theorem}
Theorem~\ref{thm-hardness} also holds even if only ${\cal R}$ is given as input and the set ${\cal L}$ is the set of all axis-parallel lines in the plane. 
The reduction of Theorem~\ref{thm-hardness} can be viewed as a more robust (and therefore gap-preserving) version of the original W[1]-hardness reduction for {\sc Rectangle Stabbing} of Giannopolous et al.~\cite{giannopoulos2013fixed}.

In Section~\ref{sec-pre} we define the notions needed for our proofs, and set up required notation. In Section~\ref{sec-algo} we prove Theorem~\ref{thm-algorithm}, while in Section~\ref{sec-hardness} we prove Theorem~\ref{thm-hardness}. We conclude with some open problems and future directions in Section~\ref{sec-conclusion}.

%\todo[inline]{maybe? discuss the variants of the problem where ${\cal L}$ is given vs when ${\cal L}$ is the set of all lines. We show hardness for the easier version (when ${\cal L}$ is all lines) and algorithm for the harder version when ${\cal L}$ is given. DL: i did it at the end of the reduction section.}

\section{Preliminaries} \label{sec-pre}
\subsection{Basic notations.}
We write $\mathbb{N} = \{1,2,\dots\}$ as the set of natural numbers, and write $\mathbb{N}_0 = \mathbb{N} \cup \{0\}$.
For a number $n \in \mathbb{N}$, we write $[n] = \{1,\dots,n\}$.
For two real numbers $a$, $b$ such that $a < b$ we denote by $(a, b)$ the open interval $\{r \in \mathbb{R} ~:~ a < r < b\}$ and by $[a, b]$ the closed interval $\{r \in \mathbb{R} ~:~ a \leq r \leq b\}$. We also define $[a, b]_\mathbb{Z} = [a, b] \cap \mathbb{Z}$.

\subsection{Strips and lines.}
Let $x$ and $y$ be real numbers. We denote by $\ell_h^y$ the \textit{horizontal line} $\mathbb{R} \times \{y\}$. 
Similarly, we denote by $\ell_v^x$ the vertical line $\{x\} \times \mathbb{R}$.
A horizontal (resp., vertical) \textit{strip}, is a region in $\mathbb{R}^2$ of the form $\mathbb{R} \times (y^-,y^+)$ (resp., $(x^-,x^+) \times \mathbb{R}$), where $y^-,y^+ \in \mathbb{R}$ (resp., $x^-,x^+ \in \mathbb{R}$) and $y^- < y^+$ (resp., $x^- < x^+$).
A \textit{strip} refers to a horizontal or vertical strip.
Two horizontal (resp., vertical) strips $P,P'$ are \textit{separated} by a horizontal (resp., vertical) line $\ell$ if $P$ and $P'$ are on opposite sides of $\ell$.
%A set $\varGamma$ of of horizontal (resp., vertical) strips are \textit{separated} by a set $\mathcal{L}$ of horizontal (resp., vertical) lines if any two (different) strips in $\varGamma$ are separated by at least one line in $\mathcal{L}$.
Let $\varGamma$ be a set of disjoint horizontal (resp., vertical) strips and $\mathcal{L}$ be a set of horizontal (resp., vertical) lines.
We say $\varGamma$ is \textit{separated} by $\mathcal{L}$ if (i) every pair of two (different) strips in $\varGamma$ are separated by at least one line in $\mathcal{L}$ and (ii) $P \cap \ell = \emptyset$ for all $P \in \varGamma$ and all $\ell \in \mathcal{L}$.
We say $\mathcal{L}$ is \textit{$\varGamma$-structured} if there exists a bijective function $f: \mathcal{L} \rightarrow \varGamma$ such that $\ell \subseteq f(\ell)$ for all $\ell \in \mathcal{L}$.
In other words, $\mathcal{L}$ is $\varGamma$-structured if every $P \in \varGamma$ contains exactly one line in $\mathcal{L}$ and these are the only lines in $\mathcal{L}$.
A set $\mathcal{L}$ of horizontal (resp., vertical) lines cut the plane into $|\mathcal{L}|+1$ horizontal (resp., vertical) strips; we use $\varGamma(\mathcal{L})$ to denote the set of these strips.
Formally, $\varGamma(\mathcal{L})$ consists of the $|\mathcal{L}|+1$ connected components of $\mathbb{R}^2 \backslash (\bigcup_{\ell \in \mathcal{L}} \ell)$.

%A horizontal (resp., vertical) \textit{closed strip} is a region in $\mathbb{R}^2$ of the form $\mathbb{R} \times [y^-,y^+]$ (resp., $[x^-,x^+] \times \mathbb{R}$), where $y^-,y^+ \in \mathbb{R}$ (resp., $x^-,x^+ \in \mathbb{R}$) and $y^- \leq y^+$ (resp., $x^- \leq x^+$). 

\subsection{Rectangles.}
All rectangles considered in this paper are axis-parallel.
Specifically, we define a rectangle as a region in $\mathbb{R}^2$ of the form $[a,b] \times [c,d]$ for $a,b,c,d \in \mathbb{R}$ such that $a \leq b$ and $c \leq d$.
We remark that all of our results easily extend to the case where all rectangles are nondegenerate, i.e. when line segments are excluded. 
A rectangle $R$ is \textit{stabbed} by a line $\ell$ if $R \cap \ell \neq \emptyset$, and is \textit{stabbed} by a set $\mathcal{L}$ of lines if it is stabbed by some $\ell \in \mathcal{L}$.
A set $\mathcal{R}$ of rectangles is \textit{stabbed} by a set $\mathcal{L}$ of lines if every $R \in \mathcal{R}$ is stabbed by $\mathcal{L}$.
For a set $\mathcal{R}$ of rectangles and a set $\mathcal{L}$ of lines, we denote by $\mathsf{opt}(\mathcal{R},\mathcal{L})$ the minimum size of a subset of $\mathcal{L}$ that stabs $\mathcal{R}$; if such a subset does not exist (i.e., $\mathcal{R}$ is not stabbed by $\mathcal{L}$), we simply define $\mathsf{opt}(\mathcal{R},\mathcal{L}) = \infty$.

\subsection{One-Dimensional Greedy Stabbing.} In the classic interval-stabbing (hitting-set) problem on the real line, we are given a collection of intervals and points, and we aim to choose the fewest points so that each interval is stabbed by at least one point.
%The optimal greedy strategy repeatedly picks the rightmost endpoint that still covers every remaining segment whose left endpoint lies to its left.
%After placing a stabbing point there, delete all segments intersecting that point, and repeat until no segments remain.
%\todo[inline]{can this be simplified to: pick the leftmost right endpoint, delete all stabbed segments, and repeat until no segments remain? We also use a subroutine \textsc{Stab} in the algorithm section; we could instead introduce that here.}
An optimal greedy strategy selects the rightmost point such that no interval lies entirely to its left, deletes all stabbed intervals, and repeats until no intervals remain.
In a {\sc Rectangle Stabbing} instance with only horizontal lines (resp. vertical lines), the same algorithm applies after projecting all rectangles onto segments on the $y$-axis (resp. $x$-axis), and projecting all lines to points on the $y$-axis (resp. $x$-axis).
Clearly this algorithm runs in polynomial time, and we will frequently use it as the subroutine \textsc{Stab}.

\section{Parameterized Approximation Algorithm} \label{sec-algo}
Consider a \textsc{Rectangle Stabbing} instance $(\mathcal{R},\mathcal{L})$ with parameter $k$.
Our goal is to compute a solution for $(\mathcal{R},\mathcal{L})$ of size at most $\frac{7}{4} k$, assuming there is a solution for $(\mathcal{R},\mathcal{L})$ of size at most $k$.
Suppose $\mathcal{L} = \mathcal{H} \cup \mathcal{V}$ where $\mathcal{H}$ (resp., $\mathcal{V}$) consists of the horizontal (resp., vertical) lines in $\mathcal{L}$.

Let $\mathcal{L}^* \subseteq \mathcal{L}$ be a (unknown) solution for $(\mathcal{R},\mathcal{L})$ with $|\mathcal{L}^*| \leq k$.
Define $\mathcal{H}^* = \mathcal{L}^* \cap \mathcal{H}$ and $\mathcal{V}^* = \mathcal{L}^* \cap \mathcal{V}$.
Suppose $k_\mathsf{h} = |\mathcal{H}^*|$ and $k_\mathsf{v} = |\mathcal{V}^*|$.
We first guess the numbers $k_\mathsf{h}$ and $k_\mathsf{v}$.
Without loss of generality, we assume that $k_\mathsf{h} \leq k_\mathsf{v}$.
We shall compute a solution of size at most $2k_\mathsf{h} + \frac{3}{2}k_\mathsf{v}$, which is sufficient since $2k_\mathsf{h} + \frac{3}{2}k_\mathsf{v} \leq \frac{7}{4} (k_\mathsf{h} + k_\mathsf{v}) \leq \frac{7}{4} k$.

Our algorithm consists of multiple steps.
During the algorithm, we shall generate various subsets $\mathcal{H}_0,\mathcal{H}_1,\mathcal{H}_1',\mathcal{H}_2 \subseteq \mathcal{H}$ and $\mathcal{V}_0,\mathcal{V}_1,\mathcal{V}_2 \subseteq \mathcal{V}$.
The sets $\mathcal{H}_1,\mathcal{H}_1',\mathcal{H}_2$ and $\mathcal{V}_1,\mathcal{V}_2$ together form our solution, while $\mathcal{H}_0$ and $\mathcal{V}_0$ are only for auxiliary uses.
Below we describe the algorithm step by step.

\subsection{Stabbing $\mathcal{R}$ using $k_\mathsf{h} + O(k^2)$ lines}
%\subsection{Greedy Horizontal Line Preselection}
Our first step is to compute $k_\mathsf{h}$ horizontal lines and $O(k^2)$ vertical lines that together stab $\mathcal{R}$ and satisfy certain nice properties.
The horizontal lines that we select in this step are denoted by ${\cal H}_1$ and will be added to the final output,
while the vertical lines selected are denoted by ${\cal V}_0$, and used as a candidate set for future selection processes.

Consider a subset $\mathcal{H}' \subseteq \mathcal{H}$ and let $y_1,\dots,y_r \in \mathbb{R}$ be the $y$-coordinates of the horizontal lines in $\mathcal{H}'$ where $y_1< \cdots < y_r$.
For convenience, set $y_0 = -\infty$.
We say $\mathcal{H}'$ is \textit{nicely positioned} with respect to another subset $\mathcal{H}'' \subseteq \mathcal{H}$ if for every $i \in [r]$, there exists a line in $\mathcal{H}''$ that is contained in the strip $\mathbb{R} \times (y_{i-1},y_i]$.
Observe that if $\mathcal{H}'$ is nicely positioned with respect to $\mathcal{H}''$, then $|\mathcal{H}'| \leq |\mathcal{H}''|$.
The main goal of the first step is given in the following lemma.

\begin{lemma} \label{lem-sweep}
    One can compute in $(|\mathcal{R}|+|\mathcal{L}|)^{O(1)}$ time $\mathcal{H}_1 \subseteq \mathcal{H}$ and $\mathcal{V}_0 \subseteq \mathcal{V}$ satisfying the following.
    \begin{enumerate}[(i)]
        \item $\mathcal{H}_1$ is nicely positioned with respect to $\mathcal{H}^*$.
        In particular, $|\mathcal{H}_1| \leq k_\mathsf{h}$.
        \item $|\mathcal{V}_0| = O(k^2)$.
        \item $\mathcal{R}$ is stabbed by $\mathcal{H}_1 \cup \mathcal{V}_0$.
    \end{enumerate}
    %$|\mathcal{H}| \leq k_\mathsf{h}$, $|\mathcal{V}| \leq k_\mathsf{h} k_\mathsf{v}$, and $\mathcal{R}$ is stabbed by $\mathcal{H} \cup \mathcal{V}$.
\end{lemma}
\begin{proof}
For two horizontal lines $h$ and $h'$, we define $\mathcal{R}(h,h') \subseteq \mathcal{R}$ as the subset consisting of all rectangles that are contained in the open strip in between $h$ and $h'$.
Suppose $\mathcal{H} = \{h_1, \dots,h_m\}$ where $h_1, \dots,h_m$ are sorted in ascending order by their y-coordinates.
For convenience, let $h_0$ (resp., $h_{m+1}$) be a horizontal line below (resp., above) all rectangles in $\mathcal{R}$.

The algorithm for computing $\mathcal{H}_1$ and $\mathcal{V}_0$ is presented in Algorithm~\ref{alg:linesweep}.
It uses the one-dimensional greedy stabbing algorithm $\textsc{Stab}({\cal R}, {\cal V})$ where we only use vertical lines in ${\cal V}$ to stab rectangles in ${\cal R}$ as a subroutine.
Starting from $h_0$, repeatedly choose the furthest line $h_{j^*}$ so that all rectangles between the current line and $h_{j^*}$ can be stabbed by at most $k_{\mathsf v}$ vertical lines; add $h_{j^*}$ to~$\mathcal H_1$ and continue from there.  Once no more $h_j$ remain, let $\mathcal R'$ be the rectangles not yet stabbed, and compute $\mathcal V_0=\textsc{Stab}(\mathcal R',\mathcal V)$ by the usual greedy process.

\begin{algorithm}[H]
        \caption{Greedy Horizontal Line Preselection}
        %\caption{\textsc{LineSweep}$(\mathcal{R,H} = \{h_1, \cdots, h_{|\mathcal{H}|}\},\mathcal{V})$}
        \begin{algorithmic}[1]
            \State $\mathcal{H}_1 \leftarrow \emptyset$
            \State $i \leftarrow 0$
            \While{$i \leq m$}
                \State $j^* \leftarrow \max\{j \in [m+1] \backslash [i]: \abs{\textsc{Stab}(\mathcal{R}(h_i,h_j),\mathcal{V})} \leq k_\mathsf{v}\}$
                \If{$j^* \leq m$}{ $\mathcal{H}_1 \leftarrow \mathcal{H}_1 \cup \{h_{j^*}\}$}
                \EndIf
                \State $i \leftarrow j^*$
            \EndWhile
            \State $\mathcal{R}' \leftarrow \{R \in \mathcal{R}: R \text{ is not stabbed by } \mathcal{H}_1\}$
            \State $\mathcal{V}_0 \leftarrow \textsc{Stab}(\mathcal{R}',\mathcal{V})$
            \State \textbf{return} $\mathcal{H}_1$ and $\mathcal{V}_0$
        \end{algorithmic}
    \label{alg:linesweep}
\end{algorithm}

%% Commented By Xiaoyang Xu - Reduce algorithm repeats on unnecessary parts
% We start with $i = 0$ and $\mathcal{H}_1 = \emptyset$.
% In each iteration, as long as $i \leq m$, we shall do the following.
% We first find the largest index $j^* > i$ satisfying that $\mathsf{opt}(\mathcal{R}(h_i,h_{j^*}),\mathcal{V}) \leq k_\mathsf{v}$.
% Note that $j^*$ always exists.
% Indeed, we have $\mathsf{opt}(\mathcal{R}(h_i,h_{i+1}),\mathcal{V}) \leq k_\mathsf{v}$, because all rectangles in $\mathcal{R}(h_i,h_{i+1})$ can only be stabbed by the lines in $\mathcal{V}$ and thus must be stabbed by $\mathcal{V}^*$.
% If $j^* \leq m$, we add $h_{j^*}$ to $\mathcal{H}_1$.
% Then we set $i = j^*$ and proceed to the next iteration.
% This procedure terminates when $i > m$, and it constructs $\mathcal{H}_1$.
% Let $\mathcal{R}' \subseteq \mathcal{R}$ consist of all rectangles that are not stabbed by $\mathcal{H}_1$.
% %Finally, we let $\mathcal{V}_0$ be a minimum subset of $\mathcal{V}$ that stabs all rectangles in $\mathcal{R}$ that are not stabbed by $\mathcal{H}_1$.
% Finally, the sub-routine $\textsc{Stab}(\mathcal{R}',\mathcal{V})$ finds a minimum subset of $\mathcal{V}$ that stabs $\mathcal{R}'$, which is the set $\mathcal{V}_0$ we want.
% As $\mathcal{V}$ only consists of vertical lines, one can implement $\textsc{Stab}(\mathcal{R}',\mathcal{V})$ in polynomial time by a simple greedy algorithm.
%% End Comment Xiaoyang Xu

% \todo[inline]{\makebox[\textwidth]{Decide if we like this placement for the figure}}

% \todo[inline]{\makebox[\textwidth]{Add figure caption}}

\begin{figure}[H]
    \centering
    \includegraphics[width=0.47\linewidth]{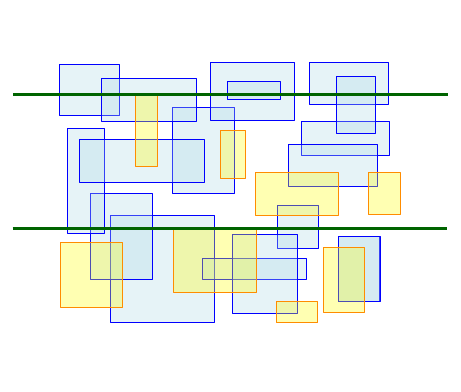}
    \hspace{0.04\linewidth}
    \includegraphics[width=0.47\linewidth]{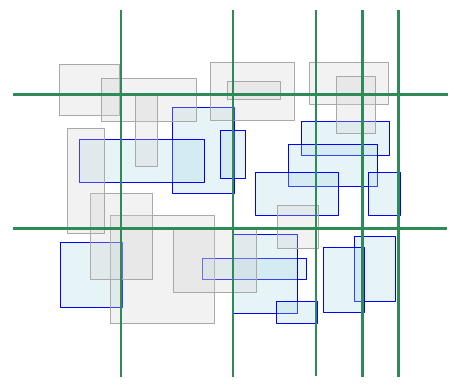}
    \caption{\textit{The set of horizontal lines $\mathcal{H}_1$ and vertical lines $\mathcal{V}_0$ computed in Lemma \ref{lem-sweep}.} On the left, sweeping upwards, we repeatedly place a horizontal line just before the number of vertically disjoint rectangles would exceed $k$. Here $k = 3$. On the right, the remaining rectangles (highlighted) after placing $\mathcal{H}_1$ are shown in blue, and stabbed by $\mathcal{V}_0$.}
    \label{fig:algorithm_1}
\end{figure}

It suffices to verify that $\mathcal{H}_1$ and $\mathcal{V}_0$ satisfy the three conditions.
Condition (iii) is clear from our construction.
To see condition (i), suppose $\mathcal{H}_1 = \{h_{i_1},\dots,h_{i_p}\}$ where $i_1< \cdots < i_p$.
Let $i_0 = 0$ and $i_{p+1} = m+1$.
We need to show that $\mathcal{H}^* \cap \{h_{i_{t-1} +1},\dots,h_{i_t}\} \neq \emptyset$ for all $t \in [p]$.
By construction, we have $\mathsf{opt}(\mathcal{R}(h_{i_{t-1}},h_{i_t +1}),\mathcal{V}) > k_\mathsf{v}$ and thus there exists $R \in \mathcal{R}(h_{i_{t-1}},h_{i_t +1})$ that is not stabbed by $\mathcal{V}^*$.
So $R$ must be stabbed by $\mathcal{H}^*$, which implies that $\mathcal{H}^*$ contains one of the lines $h_{i_{t-1} +1},\dots,h_{i_t}$.
To see condition (ii), observe that $\mathcal{R}' = \bigcup_{t=1}^{p+1} \mathcal{R}(h_{i_{t-1}},h_{i_t})$.
Since $\mathsf{opt}(\mathcal{R}(h_{i_{t-1}},h_{i_t}),\mathcal{V}) \leq k_\mathsf{v}$ for all $t \in [p+1]$ by our construction, $|\mathcal{V}_0| = \mathsf{opt}(\mathcal{R}',\mathcal{V}) \leq \sum_{t=1}^{p+1} \mathsf{opt}(\mathcal{R}(h_{i_{t-1}},h_{i_t}),\mathcal{V}) \leq k_\mathsf{v}(p+1)$.
As $p = |\mathcal{H}_1| \leq k_\mathsf{h}$, we have $|\mathcal{V}_0| = O(k^2)$.
\end{proof}

\subsection{Finding good vertical strips} \label{sec-vstrip}

We now have the sets $\mathcal{H}_1$ and $\mathcal{V}_0$ satisfying the conditions in Lemma~\ref{lem-sweep}.
The next step aims to find a set $\varGamma_\mathsf{v}\subseteq \varGamma(\mathcal{V}_0)$ of vertical strips together with a subset of vertical lines $\mathcal{V}_1 \subseteq \mathcal{V}_0$ that separates $\varGamma_\mathsf{v}$.
The main property we want is that each strip in $\varGamma_\mathsf{v}$ contains exactly one line in the optimal solution $\mathcal{V}^*$,
and these lines together with $\mathcal{H}_1 \cup \mathcal{V}_1 \cup \mathcal{H}^*$ stab $\mathcal{R}$.
Specifically, we prove the following lemma.

\begin{lemma} \label{lem-vstrip}
    There exist $\varGamma_\mathsf{v} \subseteq \varGamma(\mathcal{V}_0)$ and $\mathcal{V}_1 \subseteq \mathcal{V}_0$ satisfying the following.
    \begin{enumerate}[(i)]
        \item $|\varGamma_\mathsf{v}| + |\mathcal{V}_1| \leq \frac{3}{2} k_\mathsf{v}$.
        \item $\varGamma_\mathsf{v}$ is separated by $\mathcal{V}_1$.
        \item For each $P \in \varGamma_\mathsf{v}$, there exists exactly one line $v^*_P \in \mathcal{V}^*$ with $v_P^* \subseteq P$.
        \item $\mathcal{R}$ is stabbed by $\mathcal{H}_1 \cup \mathcal{V}_1 \cup \mathcal{H}^* \cup \{v^*_P: P \in \varGamma_\mathsf{v}\}$.
    \end{enumerate}
\end{lemma}
\begin{proof}
For each strip $P \in \varGamma(\mathcal{V}_0)$, let $\partial P \subseteq \mathcal{V}_0$ denote the (at most) two vertical lines that bound $P$.
We say $P$ is \textit{light} if $P$ contains exactly one line in $\mathcal{V}^*$ and is \textit{heavy} if $P$ contains at least two lines in $\mathcal{V}^*$.
Let $\varGamma' \subseteq \varGamma(\mathcal{V}_0)$ (resp., $\varGamma'' \subseteq \varGamma(\mathcal{V}_0)$) consist of all light (resp., heavy) strips.
Suppose $\varGamma' = \{P_1,\dots,P_r\}$ where $P_1,\dots,P_r$ are sorted from left to right.
Write $\varGamma_0' = \{P_i: i \in [r] \text{ is even}\}$ and $\varGamma_1' = \{P_i: i \in [r] \text{ is odd}\}$.
Then we define $\varGamma_\mathsf{v} = \varGamma_1'$ and $\mathcal{V}_1 = (\mathcal{V}_0 \cap \mathcal{V^*}) \cup (\bigcup_{P \in \varGamma_0'} \partial P) \cup (\bigcup_{P \in \varGamma''} \partial P)$.

\begin{figure}[h]
    \centering
    \includegraphics[width=0.47\linewidth]{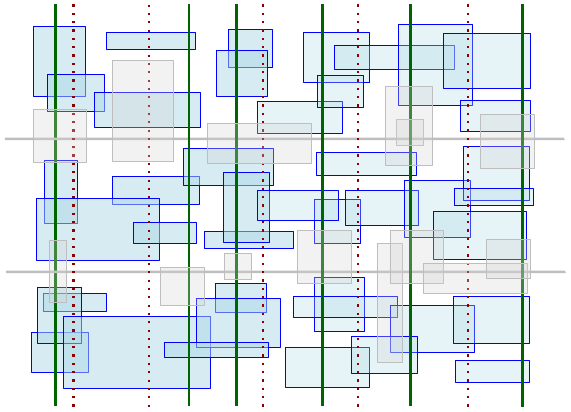}
    \hspace{0.04\linewidth}
    \includegraphics[width=0.47\linewidth]{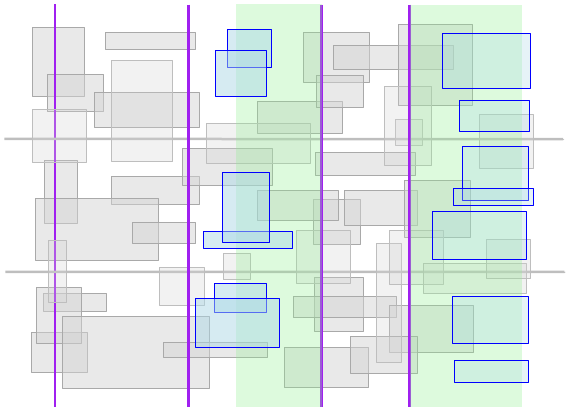}
    \caption{
    \textit{The vertical lines $\mathcal{V}_1$ and strips $\Gamma_v$ guaranteed by Lemma \ref{lem-vstrip}.}
    On the left, the vertical dashed lines represent the unknown lines of OPT, $\mathcal{V}^*$, while the solid green are those of $\mathcal{V}_0$.
    On the right, the strips of $\Gamma_v$ are highlighted green.
    The vertical lines of $\mathcal{V}_1$ in purple include the boundaries of all heavy strips, and separate the strips of $\Gamma_v$.
    }
\end{figure}

We now verify the four conditions in the lemma.
Conditions (ii) and (iii) follow readily from the construction.
To see (iii), observe that $\varGamma_\mathsf{v}$ only contains strips in $\varGamma'$ (i.e., light strips), each of which contains exactly one line in $\mathcal{V}^*$ by definition.
For (ii), consider a pair of strips $P_i, P_j \in \varGamma_\mathsf{v}$ where $i < j$.
By construction, both $i$ and $j$ are odd, which implies $j \geq i+2$.
Since $i+1$ is even, we have $P_{i+1} \in \varGamma_0'$ and thus $\partial P_{i+1} \subseteq \mathcal{V}_1$.
The lines in $\partial P_{i+1}$ then separate $P_i$ and $P_j$.

For condition (iv), let $\mathcal{R}' \subseteq \mathcal{R}$ consist of rectangles not stabbed by $\mathcal{H}_1 \cup \mathcal{H}^*$.
It suffices to show that every $R \in \mathcal{R}'$ is stabbed by $\mathcal{V}_1 \cup \{v^*_P: P \in \varGamma_\mathsf{v}\}$.
%To prove that $\mathcal{H}_1 \cup \mathcal{V}_1 \cup \mathcal{H}^* \cup \{v^*_P: P \in \varGamma_\mathsf{v}\}$ stabs $\mathcal{R}$, it suffices to show if a rectangle $r \in \mathcal{R}'$ is stabbed by some $v \in V^*$, then it is also stabbed by $\mathcal{V}_1 \cup \{v^*_P: P \in \varGamma_\mathsf{v}\}$.
Since $R$ is not stabbed by $\mathcal{H}^*$, it must be stabbed by some $v \in \mathcal{V}^*$.
If $v \in \mathcal{V}_0$, then $v \in \mathcal{V}_0 \cap \mathcal{V^*} \subseteq \mathcal{V}_1$ and thus $\mathcal{V}_1$ stabs $R$.
Otherwise, $v$ is contained in some strip $P \in \varGamma(\mathcal{V}_0)$.
Then either $P \in \varGamma_1' = \varGamma_\mathsf{v}$ or $P \in \varGamma_0' \cup \varGamma''$.
If $P \in \varGamma_\mathsf{v}$, the $v_P^* = v$ stabs $R$.
If $P \in \varGamma_0' \cup \varGamma''$, then $\partial P \subseteq \mathcal{V}_1$.
Note that at least one line in $\partial P$ stabs $R$, because $R \cap P \neq \emptyset$ and $\mathcal{V}_0$ stabs $R$.

Finally, we verify condition (i), where we need $|\varGamma_\mathsf{v}| + |\mathcal{V}_1| \leq \frac{3}{2} k_\mathsf{v}$.
We have $|\varGamma_\mathsf{v}| = |\varGamma_1'|$ and  $|\mathcal{V}_1| \leq |\mathcal{V}_0 \cap \mathcal{V^*}| + 2(|\varGamma_0'|+|\varGamma''|)$.
Note that $|\mathcal{V}_0 \cap \mathcal{V^*}| + |\varGamma'| + 2|\varGamma''| \leq |\mathcal{V^*}|$.
Therefore,
\begin{equation*}
    |\varGamma_\mathsf{v}| + |\mathcal{V}_1| \leq |\varGamma_1'| + |\mathcal{V}_0 \cap \mathcal{V^*}| + 2(|\varGamma_0'|+|\varGamma''|) = |\mathcal{V}_0 \cap \mathcal{V^*}| + |\varGamma'| + 2|\varGamma''| + |\varGamma_0'| \leq |\mathcal{V^*}| + |\varGamma_0'|.
\end{equation*}
Note that $|\varGamma_0'| \leq |\varGamma_1'|$ and thus $|\varGamma_0'| \leq \frac{1}{2} |\varGamma'| \leq \frac{1}{2} |\mathcal{V^*}|$.
So we have $|\varGamma_\mathsf{v}| + |\mathcal{V}_1| \leq \frac{3}{2} |\mathcal{V^*}| = \frac{3}{2} k_\mathsf{v}$.
\end{proof}

Since $\mathcal{H}^*$ and $\mathcal{V}^*$ are unknown to us, we cannot compute the sets $\varGamma_\mathsf{v}$ and $\mathcal{V}_1$ in the above lemma.
However, we can afford guessing them.
We have $|\varGamma(\mathcal{V}_0)| = |\mathcal{V}_0|+1 = O(k^2)$ by (ii) of Lemma~\ref{lem-sweep} and $|\varGamma_\mathsf{v}| + |\mathcal{V}_1| = O(k)$ by (i) of Lemma~\ref{lem-vstrip}.
Therefore, the total number of guesses for $\varGamma_\mathsf{v}$ and $\mathcal{V}_1$ is bounded by $k^{O(k)}$, and making the guess leads to a $k^{O(k)}$ overhead on the running time.

Note that $\{v^*_P: P \in \varGamma_\mathsf{v}\}$ is $\varGamma_\mathsf{v}$-structured.
The size of the set $\mathcal{H}_1 \cup \mathcal{V}_1 \cup \mathcal{H}^* \cup \{v^*_P: P \in \varGamma_\mathsf{v}\}$ is $|\mathcal{H}_1| + |\mathcal{V}_1| + k_\mathsf{h} + |\varGamma_\mathsf{v}|$, which is at most $2k_\mathsf{h}+\frac{3}{2} k_\mathsf{v}$ by (i) of Lemma~\ref{lem-sweep} and (i) of Lemma~\ref{lem-vstrip}.
Therefore, (iv) of Lemma~\ref{lem-vstrip} implies the existence of a solution of size at most $2k_\mathsf{h}+\frac{3}{2} k_\mathsf{v}$ that consists of $\mathcal{H}_1 \cup \mathcal{V}_1$ together with a subset of $\mathcal{H}$ (i.e., $\mathcal{H}^*$) and a subset of $\mathcal{V}$ that is $\varGamma_\mathsf{v}$-structured.

%Lemma~\ref{lem-vstrip} implies the existence of a $\varGamma_\mathsf{v}$-structured solution, i.e., the set $\mathcal{H}_1 \cup \mathcal{H}^* \cup \mathcal{V}_1 \cup \{v^*_P: P \in \varGamma_\mathsf{v}\}$ in (iv) of Lemma~\ref{lem-vstrip}.
%Indeed, (ii) of Lemma~\ref{lem-vstrip} implies that the lines in $\mathcal{V}_1$ are disjoint from the strips in $\varGamma_\mathsf{v}$.
%Thus, $\mathcal{V}_1 \cup \{v^*_P: P \in \varGamma_\mathsf{v}\}$ is $\varGamma_\mathsf{v}$-structured, so is $\mathcal{H}_1 \cup \mathcal{H}^* \cup \mathcal{V}_1 \cup \{v^*_P: P \in \varGamma_\mathsf{v}\}$.
%This $\varGamma_\mathsf{v}$-structured solution contains $\mathcal{H}_1 \cup \mathcal{V}_1$ and its size is $|\mathcal{H}_1| + |\mathcal{H}^*| + |\mathcal{V}_1| + |\varGamma_\mathsf{v}| \leq 2k_\mathsf{h}+\frac{3}{2} k_\mathsf{v}$, by (i) of Lemma~\ref{lem-sweep} and (i) of Lemma~\ref{lem-vstrip}.
%$|\varGamma(\mathcal{V}_0)| = |\mathcal{V}_0|+1 = O(k^2)$ by Lemma~\ref{lem-sweep}, we can afford guessing the sets  $\varGamma_\mathsf{v}$ and $\mathcal{V}_1$ in the above lemma.

\subsection{Removing redundant rectangles}

In what follows, we shall focus on finding a solution of size at most $2k_\mathsf{h}+\frac{3}{2} k_\mathsf{v}$ that consists of $\mathcal{H}_1 \cup \mathcal{V}_1$ together with a subset of $\mathcal{H}$ and a $\varGamma_\mathsf{v}$-structured subset of $\mathcal{V}$.
As argued at the end of the last section, such a solution exists. 
%The rest of our algorithm aims to compute a $\varGamma_\mathsf{v}$-structured solution containing $\mathcal{H}_1 \cup \mathcal{V}_1$ that is of size at most $2k_\mathsf{h}+\frac{3}{2} k_\mathsf{v}$.
%Such a solution exists by our discussion at the end of the last section.
To find such a solution, we shall first compute a good representative subset $\mathcal{K} \subseteq \mathcal{R}$ such that to solve the problem on $\mathcal{R}$, it suffices to solve the problem on $\mathcal{K}$.
The desired property of $\mathcal{K}$ is that it can be stabbed by $\mathcal{H}_1 \cup \mathcal{V}_1$ together with a set $\mathcal{H}_0 \subseteq \mathcal{H}$ of $O(k^2)$ additional horizontal lines.

\begin{lemma} \label{lem-redundant}
    One can compute in $(|\mathcal{R}|+|\mathcal{L}|)^{O(1)}$ time $\mathcal{K} \subseteq \mathcal{R}$ and $\mathcal{H}_0 \subseteq \mathcal{H}$ satisfying the following.
    \begin{enumerate}[(i)]
        \item $|\mathcal{H}_0| \leq O(k^2)$.
        \item For any $\mathcal{A} \subseteq \mathcal{H}$ with $|\mathcal{A}| \leq 2k$ and any $\varGamma_\mathsf{v}$-structured $\mathcal{B} \subseteq \mathcal{V}$, $\mathcal{R}$ is stabbed by $\mathcal{H}_1 \cup \mathcal{V}_1 \cup \mathcal{A} \cup \mathcal{B}$ if and only if $\mathcal{K}$ is stabbed by $\mathcal{H}_1 \cup \mathcal{V}_1 \cup \mathcal{A} \cup \mathcal{B}$.
        \item For each $R \in \mathcal{K}$, if $R$ is stabbed by $\mathcal{H}$, then $R$ is also stabbed by $\mathcal{H}_1 \cup \mathcal{V}_1 \cup \mathcal{H}_0$.
    \end{enumerate}
\end{lemma}

%\begin{lemma} \label{lem-redundant}
%    One can compute in $(|\mathcal{R}|+|\mathcal{L}|)^{O(1)}$ time $\mathcal{K} \subseteq \mathcal{R}$ and $\mathcal{H}_0 \subseteq \mathcal{H}$ satisfying the following.
%    \begin{enumerate}[(i)]
%        \item $|\mathcal{H}_0| \leq O(k^2)$.
%        \item For any $\mathcal{A} \subseteq \mathcal{H}$ with $|\mathcal{A}| \leq 2k$ and any $\varGamma_\mathsf{v}$-structured $\mathcal{B} \subseteq \mathcal{V}$, $\mathcal{R}$ is stabbed by $\mathcal{H}_1 \cup \mathcal{V}_1 \cup \mathcal{A} \cup \mathcal{B}$ if and only if $\mathcal{K}$ is stabbed by $\mathcal{A} \cup \mathcal{B}$.
%        \item $\mathcal{K}$ is stabbed by $\mathcal{H}_0$.
%    \end{enumerate}
%\end{lemma}
\begin{proof}
For a rectangle $R \in \mathcal{R}$ and a vertical strip $P$, we use $\mathsf{wid}_P(R)$ to denote the width of $R \cap P$, i.e., the length of the interval obtained by projecting $R \cap P$ to the $x$-axis.
For $\mathcal{R}' \subseteq \mathcal{R}$ and a line $\ell$, we use $\mathcal{R}' \Cap \ell$ to denote the subset of $\mathcal{R}'$ consisting of all rectangles stabbed by $\ell$.
As in the proof of Lemma~\ref{lem-vstrip}, for each $P \in \varGamma_\mathsf{v}$, let $\partial P \subseteq \mathcal{V}_0$ consist of the (at most) 2 borderlines of $P$.

The algorithm for computing $\mathcal{K}$ and $\mathcal{H}_0$ is presented in Algorithm~\ref{alg-redundant}.
Let $\mathcal{R}' \subseteq \mathcal{R}$ consist of the rectangles stabbed by $\mathcal{H}$ but not stabbed by $\mathcal{H}_1 \cup \mathcal{V}_1$.
The set $\mathcal{K}$ is obtained from $\mathcal{R}$ by removing a subset $\mathcal{X} \subseteq \mathcal{R}'$
as follows.
Initially, $\mathcal{X} = \emptyset$.
In each iteration, as long as there exists $\ell \in \partial P$ for some $P \in \varGamma_\mathsf{v}$ such that $\mathsf{opt}((\mathcal{R}' \backslash \mathcal{X}) \Cap \ell, \mathcal{H}) \geq 2k+2$, we add $R^*$ to $\mathcal{X}$, where $R^*$ is the rectangle in $(\mathcal{R}' \backslash \mathcal{X}) \Cap \ell$ that maximizes $\mathsf{wid}_P(R^*)$.
We keep doing this until there does not exist such a line $\ell$.
This completes the construction of $\mathcal{X}$.
We then set $\mathcal{K} = \mathcal{R} \backslash \mathcal{X}$ and define $\mathcal{H}_0 \subseteq \mathcal{H}$ as a minimum subset that stabs $\mathcal{R}' \backslash \mathcal{X}$.
Here the sub-routine $\textsc{Stab}$ is the same as the one in Algorithm~\ref{alg:linesweep}.

\begin{algorithm}
    \caption{Redundant Rectangle Elimination}
    \begin{algorithmic}[1]
        %\State $\varGamma_\mathsf{v} \leftarrow \emptyset$
        \State $\mathcal{R}' \leftarrow \{R \in \mathcal{R}: R \text{ is stabbed by } \mathcal{H} \text{ but not stabbed by } \mathcal{H}_1 \cup \mathcal{V}_1\}$
        \State $\mathcal{X} \leftarrow \emptyset$
        %\State $\mathcal{R''} \leftarrow \{r \in \mathcal{R}'| \exists P \in \varGamma_v,\ x^-(r) \leq x^-(P) \land x^+(r) \geq x^+(P)\}$
        %\State $\mathcal{R^*} \leftarrow \mathcal{R} \setminus (\mathcal{R'} \cup \mathcal{R''})$
        %\State $X^- = \{x^-(P): P\in \varGamma_v\}$
        %\State $X^+ = \{x^+(P): P\in \varGamma_v\}$
        \While{there exist $P \in \varGamma_\mathsf{v}$ and $\ell \in \partial P$ such that $\mathsf{opt}((\mathcal{R}' \backslash \mathcal{X}) \Cap \ell, \mathcal{H}) \geq 2k+2$}
            \State $R^* = \arg\max_{R \in (\mathcal{R}' \backslash \mathcal{X}) \Cap \ell} \mathsf{wid}_P(R)$
            \State $\mathcal{X} \leftarrow \mathcal{X} \cup \{R^*\}$ 
        \EndWhile
        \State $\mathcal{K} \leftarrow \mathcal{R} \backslash \mathcal{X}$
        \State $\mathcal{H}_0 \leftarrow \textsc{Stab}(\mathcal{R}' \backslash \mathcal{X}, \mathcal{H})$
        \State \textbf{return} $\mathcal{K}$ and $\mathcal{H}_0$
    \end{algorithmic}
    \label{alg-redundant}
\end{algorithm}

\begin{figure}[h]
    \centering
    \includegraphics[width=0.47\linewidth]{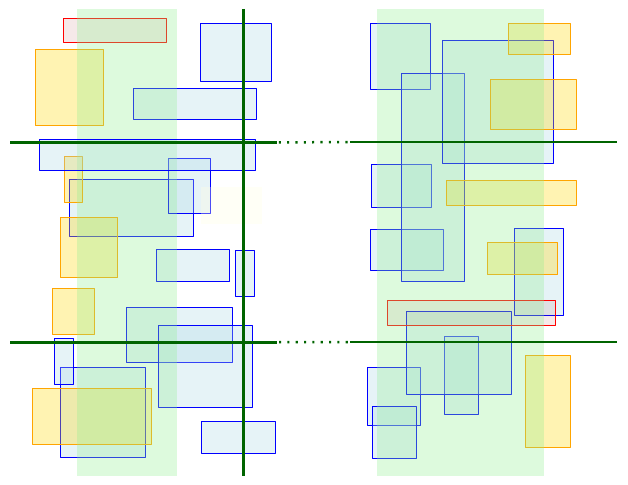}
    \hspace{0.04\linewidth}
    \includegraphics[width=0.47\linewidth]{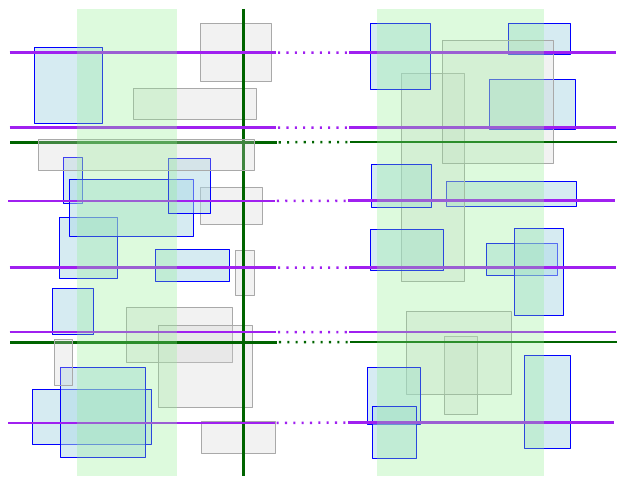}
    \caption{
    \textit{The subset of rectangles $\mathcal{K}$ and horizontal lines $\mathcal{H}_0$ computed in Lemma \ref{lem-redundant}.} On the left, for each vertical green strip of $\Gamma_v$, we remove the widest rectangle (red) that is part of a large set of horizontally disjoint rectangles (yellow) touching the strip boundary, since it must be stabbed vertically. The surviving rectangles make up $\mathcal{K}$. On the right, the set of purple lines $\mathcal{H}_0$ stab all remaining unstabbed rectangles.}
\end{figure}

We now verify the three conditions in the lemma.
Condition (iii) is clear from our construction.
Indeed, if a rectangle $R \in \mathcal{K}$ is stabbed by $\mathcal{H}$ but not stabbed by $\mathcal{H}_1 \cup \mathcal{V}_1$, then $R \in \mathcal{R}' \backslash \mathcal{X}$, which implies that $R$ is stabbed by $\mathcal{H}_0$.
Next we prove condition (i), for which we need $\mathsf{opt}(\mathcal{R}' \backslash \mathcal{X}, \mathcal{H}) = O(k^2)$.
First, observe that $\mathsf{opt}((\mathcal{R}' \backslash \mathcal{X}) \Cap \ell, \mathcal{H}) \leq 2k+1$ for all $\ell \in \bigcup_{P \in \varGamma_\mathsf{v}}\partial P$, for otherwise the while loop of the algorithm cannot terminate.
Since $|\bigcup_{P \in \varGamma_\mathsf{v}}\partial P| \leq 2|\varGamma_\mathsf{v}| = O(k)$ by (i) of Lemma~\ref{lem-vstrip}, it follows that all rectangles in $\mathcal{R}' \backslash \mathcal{X}$ stabbed by $\bigcup_{P \in \varGamma_\mathsf{v}}\partial P$ can be stabbed by $O(k^2)$ lines in $\mathcal{H}$.
So it suffices to consider the rectangles in $\mathcal{R}' \backslash \mathcal{X}$ that are not stabbed by $\bigcup_{P \in \varGamma_\mathsf{v}}\partial P$.
Note that if a rectangle $R \in \mathcal{R}' \backslash \mathcal{X}$ is not stabbed by $\partial P$, then we have $R \cap P = \emptyset$, because $P \in \varGamma_\mathsf{v} \subseteq \varGamma(\mathcal{V}_0)$ and all rectangles in $\mathcal{R}'$ are stabbed by $\mathcal{V}_0$ according to (iii) of Lemma~\ref{lem-sweep}.
Therefore, the rectangles in $\mathcal{R}' \backslash \mathcal{X}$ not stabbed by $\bigcup_{P \in \varGamma_\mathsf{v}}\partial P$ are just those disjoint from all strips in $\varGamma_\mathsf{v}$.
By (iv) of Lemma~\ref{lem-vstrip}, all such rectangles can be stabbed by $\mathcal{H}^*$.
As $|\mathcal{H}^*| \leq k$, we have condition (i) in the lemma.

Finally, we verify condition (ii).
We focus on the ``if'' direction; the ``only if'' direction is trivial because $\mathcal{K} \subseteq \mathcal{R}$.
Let $\mathcal{A} \subseteq \mathcal{H}$ with $|\mathcal{A}| \leq 2k$ and $\mathcal{B} \subseteq \mathcal{V}$ be $\varGamma_\mathsf{v}$-structured.
Since $\mathcal{R} \setminus \mathcal{K} = \mathcal{X} \subseteq \mathcal{R}'$ and no rectangle of $\mathcal{R}'$ is stabbed by $\mathcal{H}_1 \cup \mathcal{V}_1$, it suffices to show that $\mathcal{R}'$ is stabbed by $\mathcal{A} \cup \mathcal{B}$ if $\mathcal{R}' \backslash \mathcal{X}$ is stabbed by $\mathcal{A} \cup \mathcal{B}$.
%The ``only if'' direction is trivial.
Assume $\mathcal{R}' \backslash \mathcal{X}$ is stabbed by $\mathcal{A} \cup \mathcal{B}$.
Write $\mathcal{X} = \{R_1^*,\dots,R_r^*\}$ such that the rectangles are added to $\mathcal{X}$ in the order $R_r^*,\dots,R_1^*$ in Algorithm~\ref{alg-redundant}.
We show by induction that $\mathcal{R}' \setminus \{R_{i+1}^*,\dots,R_r^*\}$ is stabbed by $\mathcal{A} \cup \mathcal{B}$ for all $i \in [r]_0$, which implies that $\mathcal{R}'$ is stabbed by $\mathcal{A} \cup \mathcal{B}$.
The base case $i = 0$ clearly holds.
Suppose $\mathcal{R}' \backslash \{R_i^*,\dots,R_r^*\}$ is stabbed by $\mathcal{A} \cup \mathcal{B}$.
For induction, we need to show that $R_i^*$ is also stabbed by $\mathcal{A} \cup \mathcal{B}$.
By our assumption, at the point $R_i^*$ is added to $\mathcal{X}$, we have $\mathcal{X} = \{R_{i+1}^*,\dots,R_r^*\}$.
The reason for why we add $R_i^*$ to $\mathcal{X}$ is that $\mathsf{opt}((\mathcal{R}' \backslash \{R_{i+1}^*,\dots,R_r^*\}) \Cap \ell, \mathcal{H}) \geq 2k+2$ and $R_i^* = \arg\max_{R \in (\mathcal{R}' \backslash \{R_{i+1}^*,\dots,R_r^*\}) \Cap \ell} \mathsf{wid}_P(R)$ for some $P \in \varGamma_\mathsf{v}$ and $\ell \in \partial P$.
Note that $\mathsf{opt}((\mathcal{R}' \backslash \{R_{i+1}^*,\dots,R_r^*\}) \Cap \ell, \mathcal{H}) < \infty$, because $\mathcal{R}'$ is stabbed by $\mathcal{H}$.
Therefore, $\mathsf{opt}((\mathcal{R}' \backslash \{R_i^*,\dots,R_r^*\}) \Cap \ell, \mathcal{H}) \geq 2k+1$.
As $|\mathcal{A}| \leq 2k$ and $\mathcal{A} \subseteq \mathcal{H}$, we know that $(\mathcal{R}' \backslash \{R_i^*,\dots,R_r^*\}) \Cap \ell$ is not stabbed by $\mathcal{A}$.
But $(\mathcal{R}' \backslash \{R_i^*,\dots,R_r^*\}) \Cap \ell$ is stabbed by $\mathcal{A} \cup \mathcal{B}$ thanks to our induction hypothesis, which implies the existence of a rectangle $R \in (\mathcal{R}' \backslash \{R_i^*,\dots,R_r^*\}) \Cap \ell$ that is stabbed by a line $v \in \mathcal{B}$.
Note that $v$ must be contained in some strip in $\varGamma_\mathsf{v}$, for $\mathcal{B}$ is $\varGamma_\mathsf{v}$-structured.
We claim that $v \subseteq P$.
Since $R \in \mathcal{R}'$, $R$ is not stabbed by $\mathcal{V}_1$.
Furthermore, because $R \cap P \neq \emptyset$ and $\varGamma_\mathsf{v}$ is separated by $\mathcal{V}_1$ according to (ii) of Lemma~\ref{lem-vstrip}, we have $R \cap P' = \emptyset$ for all $P' \in \varGamma_\mathsf{v} \backslash \{P\}$.
Therefore, if $v \subseteq P'$ for some $P' \in \varGamma_\mathsf{v} \backslash \{P\}$, $R$ cannot be stabbed by $v$.
This implies $v \subseteq P$.
Using this fact, we show that $v$ stabs $R_i^*$.
As $R \in (\mathcal{R}' \backslash \{R_i^*,\dots,R_r^*\}) \Cap \ell$, $R$ intersects $\ell$ and hence one side of $R \cap P$ is bounded by $\ell$.
For the same reason, one side of $R_i^* \cap P$ is also bounded by $\ell$.
By construction, $\mathsf{wid}_P(R_i^*) \geq \mathsf{wid}_P(R)$.
Thus, if a vertical line contained in $P$ stabs $R$, it also stabs $R_i^*$.
In particular, $v$ stabs $R_i^*$.
It follows that $\mathcal{R}' \backslash \{R_{i+1}^*,\dots,R_r^*\}$ is stabbed by $\mathcal{A} \cup \mathcal{B}$.
This completes the induction argument and implies condition (i).
\end{proof}

\subsection{Finding good horizontal strips}

Based on the discussion in the previous sections, we already have the sets $\mathcal{K}$, $\mathcal{H}_0$, $\mathcal{H}_1$, $\mathcal{V}_1$, and $\varGamma_\mathsf{v}$.
The goal of the next step is somewhat similar to that in Subsection~\ref{sec-vstrip}.
This time we want to find a set $\varGamma_\mathsf{h}$ of \textit{horizontal} strips each of which contains one line in $\mathcal{H}^*$, together with a set of horizontal lines in $\mathcal{H}_0$ separating $\varGamma_\mathsf{h}$.
%We want each strip in $\varGamma_\mathsf{h}$ contains one line in $\mathcal{H}^*$ and these lines together with $\mathcal{H}_1 \cup \mathcal{H}^* \cup \mathcal{V}_1$ stab $\mathcal{R}$.
But the other desired properties for $\varGamma_\mathsf{h}$ are slightly different from those for $\varGamma_\mathsf{v}$ in Lemma~\ref{lem-vstrip}, and thus we need a different proof as well.
Recall that for each $P \in \varGamma_\mathsf{v}$, $v^*_P$ is the (unique) vertical line in $\mathcal{V}^*$ with $v_P^* \subseteq P$.
We prove the following Lemma.

\begin{lemma} \label{lem-hstrip}
    There exist $\varGamma_\mathsf{h} \subseteq \varGamma(\mathcal{H}_1 \cup \mathcal{H}_0)$ and $\mathcal{H}_1' \subseteq \mathcal{H}_0$ satisfying the following.
    \begin{enumerate}[(i)]
        %\item $|\varGamma_\mathsf{h}| \leq k_\mathsf{h}$ and $|\mathcal{H}_1'| \leq k_\mathsf{h} - |\mathcal{H}_1|$.
        \item $|\mathcal{H}_1| + |\varGamma_\mathsf{h}| + |\mathcal{H}_1'| \leq 2k_\mathsf{h}$
        \item $\varGamma_\mathsf{h}$ is separated by $\mathcal{H}_1 \cup \mathcal{H}_1'$.
        \item For each $Q \in \varGamma_\mathsf{h}$, there exists exactly one line $h^*_Q \in \mathcal{H}^*$ with $h_Q^* \subseteq Q$.
        \item $\mathcal{K}$ is stabbed by $\mathcal{H}_1 \cup \mathcal{V}_1 \cup \mathcal{H}_1' \cup \{h^*_Q: Q \in \varGamma_\mathsf{h}\} \cup \{v^*_P: P \in \varGamma_\mathsf{v}\}$.
    \end{enumerate}
\end{lemma}
\begin{proof}
%Since we are only interested in the existence of the sets $\varGamma_h, \mathcal{H}_1'$, we again start from a fixed solution $\mathcal{L^*} = \mathcal{V}^* \cup \mathcal{H}^*$, although we will only need to know $\mathcal{H}^*$.
As in the proof of Lemma \ref{lem-vstrip}, for each strip $Q \in \varGamma(\mathcal{H}_0 \cup \mathcal{H}_1)$, let $\partial Q \subseteq \mathcal{H}_0 \cup \mathcal{H}_1$ consist of the (at most) two lines adjacent to $Q$.
Also, we say a strip $Q \in \varGamma(\mathcal{H}_0 \cup \mathcal{H}_1)$ is \emph{light} if $Q$ contains exactly one line of $\mathcal{H}^*$ and \emph{heavy} if $Q$ contains at least two lines of $\mathcal{H}^*$.
Let $\varGamma' \subseteq \varGamma(\mathcal{H}_0 \cup \mathcal{H}_1)$ (resp., $\varGamma'' \subseteq \varGamma(\mathcal{H}_0 \cup \mathcal{H}_1)$) consist of the light (resp., heavy) strips.
We simply define $\varGamma_\mathsf{h} = \varGamma'$.
The set $\mathcal{H}_1' \subseteq \mathcal{H}_0$ is defined as follows.
First, we include in $\mathcal{H}_1'$ all lines in $(\mathcal{H}^* \cap \mathcal{H}_0) \cup (\bigcup_{Q \in \varGamma''} \partial Q)$.
Suppose $\varGamma' = \{Q_1,\dots,Q_t\}$, where $Q_1,\dots,Q_t$ are sorted from bottom to top.
For every $i \in [t-1]$ such that $Q_i$ and $Q_{i+1}$ are not separated by $\mathcal{H}_1 \cup (\mathcal{H}^* \cap \mathcal{H}_0) \cup (\bigcup_{Q \in \varGamma''} \partial Q)$, we add to $\mathcal{H}_1'$ an arbitrary line $h \in \mathcal{H}_0$ that separates $Q_i$ and $Q_{i+1}$; such a line exists since $Q_i,Q_{i+1} \in \varGamma(\mathcal{H}_0 \cup \mathcal{H}_1)$ and $\mathcal{H}_1$ does not separate $Q_i, Q_{i+1}$.

\medskip
\begin{figure}[H]
    \centering
    \includegraphics[width=0.47\linewidth]{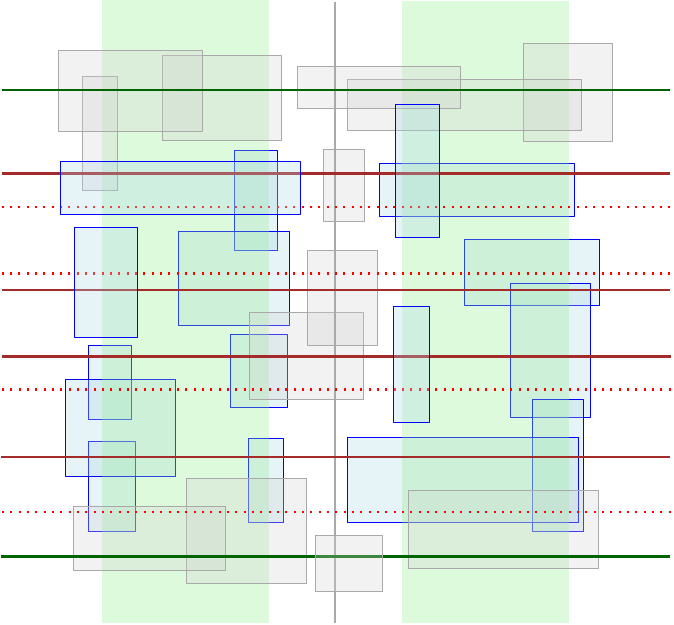}
    \hspace{0.04\linewidth}
    \includegraphics[width=0.47\linewidth]{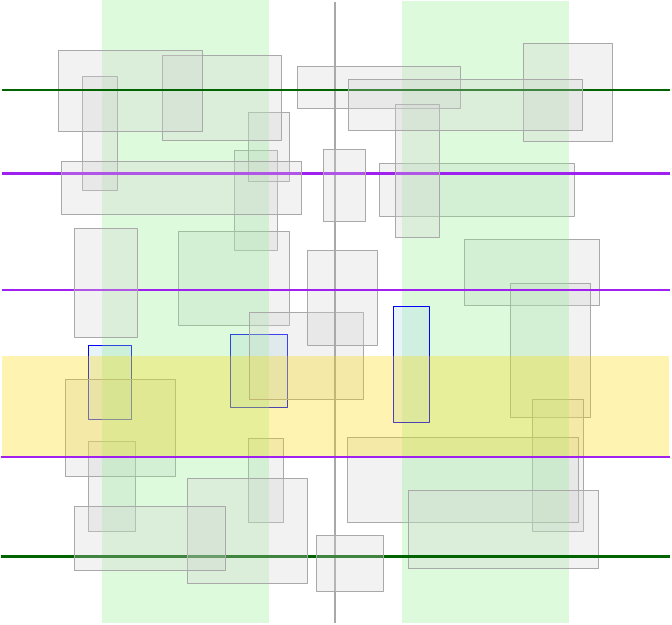}
    \caption{\textit{The horizontal lines $\mathcal{H}_1'$ and strips $\Gamma_h$ guaranteed by Lemma \ref{lem-hstrip}.} On the left, the horizontal dashed lines represent the unknown lines of OPT, $\mathcal{H}^*$, while the solid brown lines are those of $\mathcal{H}_0$ and the solid green lines are those of $\mathcal{H}_1$. On the right, the strips of $\Gamma_h$ are highlighted yellow. The horizontal lines of $\mathcal{H}_1'$ in purple include the boundaries of all heavy strips and arbitrary lines between two consecutive light strips. They separate the strips of $\Gamma_h$.
    All the rectangles are now stabbed by either a strip or a line.}
    \label{fig:algorithm_4}
\end{figure}

We now verify the four conditions in the lemma.
Conditions~(ii) and (iii) directly follow from our construction.
To see (iv), consider a rectangle $R \in \mathcal{K}$ that is not stabbed by $\mathcal{H}_1 \cup \mathcal{V}_1$.
If $R$ is not stabbed by $\mathcal{H}^*$, then (iv) of Lemma~\ref{lem-vstrip} implies that $R$ is stabbed by $\{v^*_P: P \in \varGamma_\mathsf{v}\}$.
So assume $R$ is stabbed by $\mathcal{H}^*$, and thus stabbed by $\mathcal{H}$ as well.
Then (iii) of Lemma~\ref{lem-redundant} implies that $R$ is also stabbed by $\mathcal{H}_0$.
Let $h \in \mathcal{H}^*$ be a line that stabs $R$.
If $h \in \mathcal{H}_0$, then $h \in \mathcal{H}_1'$ and hence $R$ is stabbed by $\mathcal{H}_1'$.
Otherwise, since $h \not\in \mathcal{H}_1$ because $\mathcal{H}_1$ does not stab $R$, $h \subseteq Q$ for some $Q \in \varGamma(\mathcal{H}_0 \cup \mathcal{H}_1)$.
As $\mathcal{H}_0$ stabs $R$, we know that $R$ is stabbed by $\partial Q$.
Thus, if $Q \in \varGamma''$, then $\partial Q \subseteq \mathcal{H}_1'$ and $R$ is stabbed by $\mathcal{H}_1'$.
If $Q \notin \varGamma''$, we must have $Q \in \varGamma'$ since $h \subseteq Q$.
Condition~(iii) then implies $h = h_Q^*$ and hence $R$ is stabbed by $\{h^*_Q: Q \in \varGamma_\mathsf{h}\}$.

Finally, we verify condition~(i).
We notice the inequality $|\mathcal{H}^* \cap \mathcal{H}_0| + |\varGamma'| + 2|\varGamma''| \leq |\mathcal{H}^*| = k_\mathsf{h}$.
Let $\mathcal{A} = \mathcal{H}_1' \backslash ((\mathcal{H}^* \cap \mathcal{H}_0) \cup (\bigcup_{Q \in \varGamma''} \partial Q))$.
We have 
\begin{equation*}
    |\mathcal{H}_1| + |\varGamma_\mathsf{h}| + |\mathcal{H}_1'| \leq |\mathcal{H}_1| + |\varGamma'| + (|\mathcal{H}^* \cap \mathcal{H}_0| + 2|\varGamma''| + |\mathcal{A}|) \leq k_\mathsf{h} + |\mathcal{H}_1| + |\mathcal{A}|.
\end{equation*}
Thus, to see (i), it suffices to show that $|\mathcal{H}_1| + |\mathcal{A}| \leq k_\mathsf{h}$.
To this end, we use a charging argument as follows.
Suppose $\mathcal{H}_1 = \{h_1,\dots,h_r\}$ where $h_1,\dots,h_r$ are sorted from bottom to top.
Let $y_i$ be the $y$-coordinate of $h_i$ for $i \in [r]$.
For convenience, write $y_0 = -\infty$.
For each $i \in [r]$, we charge $h_i$ to the topmost line in $\mathcal{H}^*$ whose $y$-coordinate is at most $y_i$.
By (i) of Lemma~\ref{lem-sweep}, there exists at least one line in $\mathcal{H}^*$ whose $y$-coordinate is in the range $(y_{i-1},y_i]$, and thus $h_i$ is charged to a line whose $y$-coordinate is in $(y_{i-1},y_i]$.
Therefore, the lines in $\mathcal{H}_1$ are charged to different lines in $\mathcal{H}^*$.
Next, recall how the lines in $\mathcal{A}$ are chosen.
Every line $h \in \mathcal{A}$ corresponds to an index $i \in [t-1]$ such that $Q_i$ and $Q_{i+1}$ are not separated by $\mathcal{H}_1 \cup (\mathcal{H}^* \cap \mathcal{H}_0) \cup (\bigcup_{Q \in \varGamma''} \partial Q)$; we then charge $h$ to $h_{Q_i}^*$, which is the (unique) line in $\mathcal{H}^*$ that is contained in $Q_i$.
Again, the lines in $\mathcal{A}$
% \todo{I believe this should be $\mathcal{A}$} RESOLVED, changed from $\mathcal{A}_1$ to $\mathcal{A}$
are charged to different lines in $\mathcal{H}^*$.
Now every line in $\mathcal{H}_1$ and $\mathcal{A}$ is charged to a line in $\mathcal{H}^*$.
Since $|\mathcal{H}^*| = k_\mathsf{h}$, to see $|\mathcal{H}_1| + |\mathcal{A}| \leq k_\mathsf{h}$, we only need to show that every line in $\mathcal{H}^*$ gets charged at most once.
As aforementioned, the lines in $\mathcal{H}_1$ (resp., $\mathcal{A}$) are charged to different lines in $\mathcal{H}^*$.
Thus, it suffices to exclude the case where a line $h \in \mathcal{H}^*$ gets charged by both $\mathcal{H}_1$ and $\mathcal{A}$.
If the $y$-coordinate of $h$ is greater than $y_r$, then only the lines in $\mathcal{A}$ can be charged to $h$.
So suppose the $y$-coordinate of $h$ is in the range $(y_{i-1},y_i]$ for $i \in [r]$.
If $h$ is not the topmost line in $\mathcal{H}^*$ with $y$-coordinate in $(y_{i-1},y_i]$, then again only the lines in $\mathcal{A}$ can be charged to $h$.
So suppose $h$ is the topmost line in $\mathcal{H}^*$ with $y$-coordinate in $(y_{i-1},y_i]$.
We claim that no line in $\mathcal{A}$ is charged to $h$.
Assume there is a line in $\mathcal{A}$ charged to $h$, which is chosen to separate the strips $Q_j$ and $Q_{j+1}$.
We then have $h = h_{Q_j}^*$.
Furthermore, $Q_j$ and $Q_{j+1}$ are not separated by $\mathcal{H}_1 \cup (\mathcal{H}^* \cap \mathcal{H}_0) \cup (\bigcup_{Q \in \varGamma''} \partial Q)$, and in particular not separated by $\mathcal{H}_1$.
It follows that the $y$-coordinate of $h_{Q_{j+1}}^*$ is also in $(y_{i-1},y_i]$.
Since $h_{Q_{j+1}}^* \in \mathcal{H}^*$ and $h_{Q_{j+1}}^*$ is above $h_{Q_j}^*$, this contradicts the assumption that $h$ is the topmost line in $\mathcal{H}^*$ with $y$-coordinate in $(y_{i-1},y_i]$.
Therefore, no line in $\mathcal{A}$ is charged to $h$.
We can finally conclude that $|\mathcal{H}_1| + |\mathcal{A}| \leq k_\mathsf{h}$, which implies condition~(i).
\end{proof}

Again, we cannot compute the sets $\varGamma_\mathsf{h}$ and $\mathcal{H}_1'$ in the above lemma, but we can afford guessing them.
We have $|\varGamma(\mathcal{H}_1 \cup \mathcal{H}_0)| = |\mathcal{H}_1 \cup \mathcal{H}_0|+1 = O(k^2)$ by (i) of Lemma~\ref{lem-sweep} and (i) of Lemma~\ref{lem-redundant}.
Also, $|\varGamma_\mathsf{h}| + |\mathcal{H}_1'| = O(k)$ by (i) of Lemma~\ref{lem-hstrip}.
Therefore, the total number of guesses for $\varGamma_\mathsf{h}$ and $\mathcal{H}_1'$ is bounded by $k^{O(k)}$, and making the guess leads to another $k^{O(k)}$ overhead on the running time.

%We say a subset $\mathcal{L}' \subseteq \mathcal{L}$ is \textit{$(\varGamma_\mathsf{h},\varGamma_\mathsf{v})$-structured} if for each $P \in \varGamma_\mathsf{h} \cup \varGamma_\mathsf{v}$, there exists exactly one line $\ell \in \mathcal{L}'$ with $\ell \subseteq P$.
%Clearly, if $\mathcal{L}'$ is $(\varGamma_\mathsf{h},\varGamma_\mathsf{v})$-structured, it is also $\varGamma_\mathsf{v}$-structured.
%Note that the set $\mathcal{H}_1 \cup \mathcal{H}_1' \cup \{h^*_Q: Q \in \varGamma_\mathsf{h}\} \cup \mathcal{V}_1 \cup \{v^*_P: P \in \varGamma_\mathsf{v}\}$ in (iv) of Lemma~\ref{lem-hstrip} is $(\varGamma_\mathsf{h},\varGamma_\mathsf{v})$-structured.
%Indeed, (ii) of Lemma~\ref{lem-hstrip} implies that the lines in $\mathcal{H}_1 \cup \mathcal{H}_1'$ are disjoint from the strips in $\varGamma_\mathsf{h}$, and as argued before, $\mathcal{V}_1 \cup \{v^*_P: P \in \varGamma_\mathsf{v}\}$ is $\varGamma_\mathsf{v}$-structured.
%The size of $\mathcal{H}_1 \cup \mathcal{H}_1' \cup \{h^*_Q: Q \in \varGamma_\mathsf{h}\} \cup \mathcal{V}_1 \cup \{v^*_P: P \in \varGamma_\mathsf{v}\}$ is $|\mathcal{H}_1| + |\mathcal{H}_1'| + |\varGamma_\mathsf{h}| + |\mathcal{V}_1| + |\varGamma_\mathsf{v}| \leq 2k_\mathsf{h}+\frac{3}{2} k_\mathsf{v}$, by (i) of Lemma~\ref{lem-sweep}, (i) of Lemma~\ref{lem-vstrip}, and (i) of Lemma~\ref{lem-hstrip}.

\subsection{Solving the problem by reducing to 2-SAT}
Note that (iv) of Lemma~\ref{lem-hstrip} implies that $\mathcal{K}$ can be stabbed by $\mathcal{H}_1 \cup \mathcal{V}_1 \cup \mathcal{H}_1'$ together with a $\varGamma_\mathsf{h}$-structured subset of $\mathcal{H}$ and a $\varGamma_\mathsf{v}$-structured subset of $\mathcal{V}$.
Thus, the last step of our algorithm is just to compute a $\varGamma_\mathsf{h}$-structured $\mathcal{H}_2 \subseteq \mathcal{H}$ and a $\varGamma_\mathsf{v}$-structured $\mathcal{V}_2 \subseteq \mathcal{V}$ such that $\mathcal{H}_1 \cup \mathcal{V}_1 \cup \mathcal{H}_1' \cup \mathcal{H}_2 \cup \mathcal{V}_2$ stabs $\mathcal{K}$. 
The size of $\mathcal{H}_1 \cup \mathcal{V}_1 \cup \mathcal{H}_1' \cup \mathcal{H}_2 \cup \mathcal{V}_2$ is $|\mathcal{H}_1| + |\mathcal{V}_1| + |\mathcal{H}_1'| + |\varGamma_\mathsf{h}| + |\varGamma_\mathsf{v}|$, which is at most $2k_\mathsf{h}+\frac{3}{2} k_\mathsf{v}$ by (i) of Lemma~\ref{lem-vstrip} and (i) of Lemma~\ref{lem-hstrip}.
%The total number of lines in these sets is $|\mathcal{H}_1| + |\mathcal{V}_1| + |\mathcal{H}_1'| + |\varGamma_\mathsf{h}| + |\varGamma_\mathsf{v}|$, which is at most $2k_\mathsf{h}+\frac{3}{2} k_\mathsf{v}$ by (i) of Lemma~\ref{lem-sweep}, (i) of Lemma~\ref{lem-vstrip}, and (i) of Lemma~\ref{lem-hstrip}.

Let $\mathcal{K}' \subseteq \mathcal{K}$ consist of the rectangles \textit{not} stabbed by $\mathcal{H}_1 \cup \mathcal{V}_1 \cup \mathcal{H}_1'$.
%Then we want $\mathcal{H}_2 \cup \mathcal{V}_2$ stabs $\mathcal{K}'$.
For each strip $Q \in \varGamma_\mathsf{h}$ (resp., $P \in \varGamma_\mathsf{v}$), define $\mathcal{H}_Q \subseteq \mathcal{H}$ (resp., $\mathcal{V}_P \subseteq \mathcal{V}$) as the subset consisting of all lines that are contained in $Q$ (resp., $P$).
Then our task is to pick one line from $\mathcal{H}_Q$ for each $Q \in \varGamma_\mathsf{h}$ (which forms $\mathcal{H}_2$) and pick one line from $\mathcal{V}_P$ for each $P \in \varGamma_\mathsf{v}$ (which forms $\mathcal{V}_2$) such that $\mathcal{K}'$ is stabbed by the picked lines.
We show that this task can be reduced to solving a 2-SAT instance, which can be done in polynomial time. The key observation is the following.

\begin{lemma} \label{lem-atmost1}
    Each rectangle in $\mathcal{K}'$ intersects at most one strip in $\varGamma_\mathsf{h}$ and at most one strip in $\varGamma_\mathsf{v}$.
\end{lemma}
\begin{proof}
Let $R \in \mathcal{K}'$.
Assume $R$ intersects two (different) strips $P, P' \in \varGamma_\mathsf{v}$.
By (ii) of Lemma~\ref{lem-vstrip}, $\mathcal{V}_1$ separates $\varGamma_\mathsf{v}$ and in particular separates $P$ and $P'$.
Thus, $\mathcal{V}_1$ stabs $R$, contradicting the fact that $R \in \mathcal{K}'$.
Similarly, assume $R$ intersects two (different) strips $Q, Q' \in \varGamma_\mathsf{h}$.
By (ii) of Lemma~\ref{lem-hstrip}, $\mathcal{H}_1 \cup \mathcal{H}_1'$ separates $\varGamma_\mathsf{h}$ and in particular separates $Q$ and $Q'$.
Thus, $\mathcal{H}_1 \cup \mathcal{H}_1'$ stabs $R$, contradicting the fact that $R \in \mathcal{K}'$.
\end{proof}
Using Lemma \ref{lem-atmost1}, it is not hard to reduce the remaining task to a 2-SAT instance, which is polynomial time solvable.
We get the following Lemma.

\begin{lemma}
\label{lem-2sat}
    One can compute in $(|\mathcal{R}|+|\mathcal{L}|)^{O(1)}$ time a $\varGamma_\mathsf{h}$-structured $\mathcal{H}_2 \subseteq \mathcal{H}$ and a $\varGamma_\mathsf{v}$-structured $\mathcal{V}_2 \subseteq \mathcal{V}$ such that $\mathcal{H}_1 \cup \mathcal{V}_1 \cup \mathcal{H}_1' \cup \mathcal{H}_2 \cup \mathcal{V}_2$ stabs $\mathcal{K}$.
\end{lemma}
\begin{proof}
    The variables of the 2-SAT instance are defined as follows.
For each $P \in \varGamma_\mathsf{v}$ and each $v \in \mathcal{V}_P$, we define a variable $x_{P,\geq v}$, which is used to indicate whether the line we select from $\mathcal{V}_P$ is to the right of $v$ (including $v$ itself).
Similarly, for each $Q \in \varGamma_\mathsf{h}$ and each $h \in \mathcal{H}_Q$, we define a variable $y_{Q,\geq h}$, which is used to indicate whether the line we select from $\mathcal{H}_Q$ is above $h$ (including $h$ itself).
We need two types of clauses.
The first type of clauses ensure that assignment to the variables is consistent in the sense that for each $P \in \varGamma_\mathsf{v}$ (resp., $Q \in \varGamma_\mathsf{h}$), the variables $x_{P,\geq v}$ (resp., $y_{Q,\geq h}$) correctly encodes the choice of one line in $\mathcal{V}_P$ (resp., $\mathcal{H}_Q$).
Specifically, for each $P \in \varGamma_\mathsf{v}$ and two lines $v,v' \in \mathcal{V}_P$ such that $v'$ is to the right of $v$, we introduce a $2$-literal clause $x_{P,\geq v'} \rightarrow x_{P,\geq v}$, i.e., $(\neg x_{P,\geq v'}) \vee x_{P,\geq v}$.
Also, for each $Q \in \varGamma_\mathsf{h}$ and two lines $h,h' \in \mathcal{H}_Q$ such that $h'$ is above $h$, we introduce a $2$-literal clause $y_{Q,\geq h'} \rightarrow y_{Q,\geq h}$, i.e., $(\neg y_{Q,\geq h'}) \vee y_{Q,\geq h}$.
If an assignment to the variables satisfies these clauses, then the lines chosen from $\mathcal{V}_P$ (resp., $\mathcal{H}_Q$) is just the rightmost $v \in \mathcal{V}_P$ (resp., topmost $h \in \mathcal{H}_Q$) such that the variable $x_{P,\geq v}$ (resp., $y_{Q,\geq h}$) is true.

The second type of clauses ensure that each rectangle in $\mathcal{K}'$ is stabbed by the lines chosen by the variables.
Consider a rectangle $R \in \mathcal{K}'$.
We shall introduce $O(1)$ clauses for $R$ such that the lines chosen stab $R$ if and only if the clauses are satisfied.
We know that $R$ intersects at least one strip in $\varGamma_\mathsf{h} \cup \varGamma_\mathsf{v}$, because the rectangles in $\mathcal{K}'$ can be stabbed by a $\varGamma_\mathsf{h}$-structured subset of $\mathcal{H}$ together with a $\varGamma_\mathsf{v}$-structured subset of $\mathcal{V}$.
We consider three cases: (i) $R$ is disjoint from the strips in $\varGamma_\mathsf{h}$, (ii) $R$ is disjoint from the strips in $\varGamma_\mathsf{v}$, and (iii) $R$ intersects both strips in $\varGamma_\mathsf{h}$ and strips in $\varGamma_\mathsf{v}$.
The first two cases are symmetric, and thus it suffices to consider case (i).
Suppose $R$ is disjoint from the strips in $\varGamma_\mathsf{h}$.
Then by Lemma~\ref{lem-atmost1}, there exists a unique $P \in \varGamma_\mathsf{v}$ with $R \cap P \neq \emptyset$, and $R$ can only be stabbed by the vertical line chosen from $\mathcal{V}_P$.
Let $v \in \mathcal{V}_P$ be the leftmost line that stabs $R$, and $v' \in \mathcal{V}_P$ be the leftmost line to the right of $v$ that does not stab $R$.
We then introduce two $1$-literal clauses $x_{P,\geq v}$ and $\neg x_{P,\geq v'}$; clearly, $R$ is stabbed by the line chosen from $\mathcal{V}_P$ if and only if both of $x_{P,\geq v}$ and $\neg x_{P,\geq v'}$ are true.
(If the line $v'$ does not exist, then we only need one clause $x_{P,\geq v}$.)
Next, we consider case (iii).
In this case, by Lemma~\ref{lem-atmost1}, there exists a unique $P \in \varGamma_\mathsf{v}$ with $R \cap P \neq \emptyset$ and a unique $Q \in \varGamma_\mathsf{h}$ with $R \cap Q \neq \emptyset$.
The rectangle $R$ can be stabbed by the line chosen from $\mathcal{V}_P$ or the line chosen from $\mathcal{H}_Q$.
Let $v \in \mathcal{V}_P$ (resp., $h \in \mathcal{H}_Q$) be the leftmost (resp., bottommost) line that stabs $R$, and $v' \in \mathcal{V}_P$  (resp., $h' \in \mathcal{H}_Q$) be the leftmost (resp., bottommost) line to the right of $v$ (resp., above $h$) that does not stab $R$.
Then $R$ is stabbed if and only if $(x_{P,\geq v} \wedge (\neg x_{P,\geq v'})) \vee (y_{Q,\geq h} \wedge (\neg y_{Q,\geq h'}))$, which is equivalent to the conjunction of the four $2$-literal clauses $x_{P,\geq v} \vee y_{Q,\geq h}$, $x_{P,\geq v} \vee (\neg y_{Q,\geq h'})$, $(\neg x_{P,\geq v'}) \vee y_{Q,\geq h}$, and $(\neg x_{P,\geq v'}) \vee (\neg y_{Q,\geq h'})$.

Based on the above discussion, we can construct a 2-CNF $\phi$ of polynomial length solving which gives us a $\varGamma_\mathsf{h}$-structured subset of $\mathcal{H}$ and a $\varGamma_\mathsf{v}$-structured subset of $\mathcal{V}$ which together stab $\mathcal{K}'$.
\end{proof}

Now we are ready to prove Theorem~\ref{thm-algorithm}.
We already computed $\mathcal{H}_1 \cup \mathcal{V}_1 \cup \mathcal{H}_1' \cup \mathcal{H}_2 \cup \mathcal{V}_2$ that stabs $\mathcal{K}$, and the time cost is $k^{O(k)} n^{O(1)}$ where $n = |\mathcal{R}|+|\mathcal{L}|$.
Indeed, the guess of $\varGamma_\mathsf{v}$, $\varGamma_\mathsf{h}$, $\mathcal{V}_1$, $\mathcal{H}_1'$ results in the $k^{O(k)}$ factor, and all the other steps can be done in polynomial time.
Note that
\begin{align*}
    |\mathcal{H}_1 \cup \mathcal{V}_1 \cup \mathcal{H}_1' \cup \mathcal{H}_2 \cup \mathcal{V}_2| =\ & |\mathcal{H}_1| + |\mathcal{V}_1| + |\mathcal{H}_1'| + |\varGamma_\mathsf{v}| + |\varGamma_\mathsf{h}| \\
    \leq\ & (|\mathcal{H}_1| + |\mathcal{H}_1'| + |\varGamma_\mathsf{h}|) + (|\mathcal{V}_1| + |\varGamma_\mathsf{v}|) \\
    \leq\ & 2k_\mathsf{h} + \frac{3}{2}k_\mathsf{v} \leq \frac{7}{4} k,
\end{align*}
by (i) of Lemma~\ref{lem-vstrip} and (i) of Lemma~\ref{lem-hstrip}.
Finally, it suffices to show that $\mathcal{H}_1 \cup \mathcal{V}_1 \cup \mathcal{H}_1' \cup \mathcal{H}_2 \cup \mathcal{V}_2$ stabs $\mathcal{R}$.
Write $\mathcal{A} = \mathcal{H}_1' \cup \mathcal{H}_2$ and $\mathcal{B} = \mathcal{V}_2$.
We have $|\mathcal{A}| \leq 2k$ and $\mathcal{B}$ is $\varGamma_\mathsf{v}$-structured.
As $\mathcal{H}_1 \cup \mathcal{V}_1 \cup \mathcal{H}_1' \cup \mathcal{H}_2 \cup \mathcal{V}_2$ stabs $\mathcal{K}$, (ii) of Lemma~\ref{lem-redundant} implies that it also stabs $\mathcal{R}$.
This completes the proof of Theorem~\ref{thm-algorithm}.

% \section{Rectangle Stabbing in 3 Dimensions}
% \label{sec-alg3d}
% \input{alg3d}

\section{Hardness of Parameterized Approximation} \label{sec-hardness}
%\begin{theorem}\label{thm:hardness}
%    Assuming $W[1]\neq FPT$, for all constant $\epsilon>0$, there is no $FPT$-time algorithm such that, on every input instance $I$ of \textsc{Rectangle Stabbing} with parameter $k$:
    %\begin{itemize}
        %\item Output ``yes'' if there is a solution using at most $k$ lines;
        %\item Output ``no'' if every solution uses at least $(1.25-\epsilon)k$ edges.
    %\end{itemize}
%\end{theorem}

In this section, we prove Theorem~\ref{thm-hardness} by a gap preserving reduction from the \textsc{Multicolored Clique} problem, which is defined as follows. 
\begin{tcolorbox}[colback=gray!5!white,colframe=gray!75!black]
        \textsc{Multicolored Clique} \hfill \textbf{Parameter:} $k$
        \vspace{0.1cm} \\
        \textbf{Input:} An undirected graph $G$; A partition $\mathcal{P}=V_1,V_2,V_3,\dots,V_k$ for $V(G)$ where $|V_1|=|V_2|=\dots =|V_k|=r$. \\
        \textbf{Goal:} Find the maximum size clique $C$ such that $\forall i\in [k],|C\cap V_i|\leq 1$. 
\end{tcolorbox}
A clique $C$ is said to be {\em multicolored}, if for all $i\in [k]$ we have $|C\cap V_i|\leq 1$. We will rely on the following result of Chen et al.~\cite{chen2025simple} (see also~\cite{karthik2022almost,lin2021constant,LinRSW22}).

\begin{proposition}\label{prop:cliqueHardness}
Assuming FPT $\neq$ W[1],
for every function $h : \mathbb{N} \rightarrow \mathbb{N}$ such that $h(k) = k^{o(1)}$
there is no algorithm that takes as input a graph $G$ and integer $k$, runs in time $f(k)|V(G)|^{O(1)}$ and either concludes that $G$ has no clique of size at least $k$ or produces a clique of size at least $k/h(k)$.
\end{proposition}

Proposition~\ref{prop:cliqueHardness} asserts hardness of the {\sc Clique} problem, while we need hardness of the {\sc Multicolored Clique} problem.  Theorem 13.7 of~\cite{cygan2015parameterized} gives a parameterized reduction from {\sc Clique} to {\sc Multicolored Clique}. The proof of this theorem is an algorithm that given $G, k$ constructs a graph $G'$ with partition $V_1, \ldots, V_k$ of $V(G')$, such that {\em (i)} if $G$ has clique of size $k$ then $G'$ has a multicolored clique of size $k$, and {\em (ii)} given a multi-colored clique $C'$ in $G'$ with $|C'|=k$ one can construct in polynomial time a clique $C$ in $G$ with $|C|=|C'|$. The very short proof of {\em (ii)} never uses the assumption that $|C'|=k$, so it in fact shows that {\em (ii)'} given a multi-colored clique $C'$ in $G'$ 
% \todo{We shouldn't again have $|C'| = k$? Copy paste error? (I think i agree - AM, if someone else also does them remove it\\ {\em I believe this to be resolved - Ajay} (RESOLVED-ajay&anikait}
one can construct in polynomial time a clique $C$ in $G$ with $|C|=|C'|$. The construction of $G'$ from $G$ in Theorem 13.7 of~\cite{cygan2015parameterized} together with {\em (i)} and {\em (ii)'} imply the following. 

\begin{theorem}\label{prop:mccHardness}
Assuming FPT $\neq$ W[1],
for every function $h : \mathbb{N} \rightarrow \mathbb{N}$ such that $h(k) = k^{o(1)}$
there is no algorithm that takes as input a graph $G$, integer $k$ and partition $V_1, V_2, \ldots V_k$ of $V(G)$, runs in time $f(k)|V(G)|^{O(1)}$ and either concludes that $G$ has no multicolored clique of size at least $k$ or produces a multicolored clique of size at least $k/h(k)$.
\end{theorem}

%We aim to prove the following lemma, thus proving Theorem~\ref{thm:hardness}.

%%% Turning off for now. Not sure we need a sketch given how short the reduction is. Might want to add some of the ideas into the introduction, through. 
%\input{sketch_hardness}
%\subsection{The reduction}
%\todo[inline]{the reduction allows degenerate rectangles that are just one line, this is not really necessary but slightly simplifies notation. Maybe make this as a remark.}

%\todo[inline]{It appears that the algorithm has open strips and the reduction has closed strips. I think in the preliminaries we should say strip means open strip and closed strip is a closed strip, and then fix the notation in the reduction. DL: DONE}

\subsection{Reduction from {\sc Multicolored Clique} to {\sc Rectangle Stabbing}}\label{sec:construction}

We now present a reduction that given a graph $G$, integers $k$ and $r$ and a partition of $V(G)$ into $V_1,V_2,\dots,V_k$ such that $|V_i|=r$ for every $i \in [k]$ produces an instance of {\sc Rectangle Stabbing} in polynomial time. 
Without loss of generality, for every $i\in [k]$ we have that $V_i = \{v_1^i,v_2^i,\dots,v_r^i\}$.
We will construct an instance $(\mathcal{R},\mathcal{L})$ of \textsc{Rectangle Stabbing} with parameter $4k$. 
$\mathcal{L} = \{\ell_h^y ~:~ y \in \mathbb{Z}\} \cup \{\ell_v^x ~:~ x \in \mathbb{Z}\}$. In other words, $\mathcal{L}$ consists of all lines at integer positions. 
%the integer coordinate horizontal lines and vertical lines in the area $[2r+1,2kr+2r]\times [2r+1,2kr+2r]$ (including the boundaries). 
%
%$\mathcal{L}$ consists of all the integer coordinate horizontal lines and vertical lines in the area $[2r+1,2kr+2r]\times [2r+1,2kr+2r]$ (including the boundaries). 
%\todo{is ${\cal L}$ the set of all integer lines? (Trary) I define as lines with non-negative positions (to show by $\mathcal{R}_F$, each interval is stabbed by exactly one line more easily), but if we allow degenerated rectangles then all integer lines work.}
The set ${\cal R}$ of rectangles consists of three parts $\mathcal{R}_F,\mathcal{R}_A,\mathcal{R}_E$, that is 
\(
    \mathcal{R}=\mathcal{R}_F\cup \mathcal{R}_A \cup \mathcal{R}_E.
\)

%\todo[inline]{what is the point of the $+2r$ offset in the strip definition?? (appears to be to avoid having an offset in $R_A$, and also makes some offsets in the last proof nicer) - keeping it.}
Before constructing each of the sets $\mathcal{R}_F,\mathcal{R}_A,\mathcal{R}_E$ we define a set of $4k$ strips, and briefly discuss the intuition behind the reduction.
It is worth pointing out that in contrast with the strips used in the algorithm section, here the strips are \textit{closed} in the sense that they include their bounding lines. 
For every $x\in [0,2k-1]_{\mathbb{Z}}$ we define the vertical closed strip $\zeta_v^x = [2r+xr+1,2r+xr+r]\times \mathbb{R}$
and the horizontal closed strip $\zeta_h^x = \mathbb{R}\times [2r+xr+1,2r+xr+r]$.
Since these are the only (closed or open) strips discussed in this reduction, we will simply refer to the closed strips above as strips.

The intention of the construction is that a solution of size $4k$ selects precisely one line from each of the above strips. This will be ensured by the rectangles in $\mathcal{R}_F$.
For each $i \in [k]$, each of the solution lines in the strips $\zeta_v^{2i-2}$, $\zeta_v^{2i-1}$, $\zeta_h^{2i-2}$ and $\zeta_h^{2i-1}$ encodes the choice of a vertex in $V_i$ in the multicolored clique.
For each $i \in [k]$ we will add rectangles in the squares
$\zeta_v^{2i-2} \cap \zeta_h^{2i-2}$, $\zeta_v^{2i-2} \cap \zeta_h^{2i-1}$, $\zeta_v^{2i-1} \cap \zeta_h^{2i-2}$ and $\zeta_v^{2i-1} \cap \zeta_h^{2i-1}$
%
%$(\zeta_v^{2i-2} \cup \zeta_v^{2i-1}) \cap (\zeta_h^{2i-2} \cup \zeta_h^{2i-1})$ 
that ensure that the four solution lines in the strips $\zeta_v^{2i-2}$, $\zeta_v^{2i-1}$, $\zeta_h^{2i-2}$ and $\zeta_h^{2i-1}$ all encode the {\em same} vertex in $V_i$. The set of these rectangles is $\mathcal{R}_E$.
Finally, for every ordered pair $i, j \in [k]$ with $i \neq j$ we add rectangles in the square\footnote{Technically, this set is not a square, but a square with a cross cut out of it. This is an imprecision for the sake of simplifying notation in the intuitive explanation of the reduction.} $(\zeta_v^{2i-2} \cup \zeta_v^{2i-1}) \cap (\zeta_h^{2j-2} \cup \zeta_h^{2j-1})$ that ensure that, as long as the solution lines in $(\zeta_v^{2i-2}$ and $\zeta_v^{2i-1})$ select the same vertex in  $V_i$, and the solution lines in $(\zeta_h^{2j-2}$ and $\zeta_h^{2j-1})$ select the same vertex in  $V_j$, then the vertex selected in $V_i$ and the vertex selected in $V_j$ are adjacent in $G$. This set of rectangles is $\mathcal{R}_A$. 

\subsection{Precise Construction of the Hardness Reduction}

%We now describe the set ${\cal R}$ of rectangles, which consists of three parts $\mathcal{R}_F,%\mathcal{R}_A,\mathcal{R}_E$, that is 
%\[
    %\mathcal{R}=\mathcal{R}_F\cup \mathcal{R}_A \cup \mathcal{R}_E
%\]
%In the following, we describe the construction of each of the sets $\mathcal{R}_F,\mathcal{R}_A$ and $\mathcal{R}_E$.

We first describe $\mathcal{R}_F$, the set of rectangles that {\bf f}orce at at least one line in each of the $4k$ strips. 
%This set of rectangles ensures that every solution of size exactly $4k$ contains for each $x\in [0,2k-1]_{\mathbb{Z}}$ exactly one vertical line in the vertical strip $\zeta_v^x$ and exactly one horizontal line in the horizontal strip $\zeta_h^x$ (this is not entirely true, but it is morally true, see the statement of Lemma~\ref{}).
%

For every $x \in [0,2k-1]_\mathbb{Z}$ and every $q \in [-5k, -1]_{\mathbb{Z}}$ we define the rectangles 
\[F_v^{x,q} = [2r+xr+1,2r+xr+r]\times \{q\},\]
\[F_h^{x,q} = \{q\} \times [2r+xr+1,2r+xr+r].\]
We set 
\[\mathcal{R}_F = \{F_v^{x,q}, F_h^{x,q} ~:~ x \in [0,2k-1]_\mathbb{Z}, ~q \in [-5k, -1]_{\mathbb{Z}}\}.\]
%\cup \{F_h^x ~:~ x \in [0,2k-1]_\mathbb{Z}\}$

%$\mathbb{R}\times [2r+xr+1,2r+xr+r]$, for $x\in [0,2k-1]\cap \mathbb{N}$. We set 
%\[
%\begin{array}{rl}
%    \mathcal{R}_F = & \{[2r+xr+1,2r+xr+r]\times [-2x-2,-2x-1]:x\in [0,2k-1]\cap \mathbb{N}\} \\
%     & \cup\{[-2x-2,-2x-1]\times [2r+xr+1,2r+xr+r]:x\in [0,2k-1]\cap \mathbb{N}\}
%\end{array}
%\]

Next we describe $\mathcal{R}_A$, the rectangles encoding the {\bf a}djacency matrix of $G$. 
%This set of rectangles is to encode the adjacency matrix. 
For every $i,j\in [k]$ such that $i\neq j$, for every $p,q\in [r]$ such that $(v_p^i,v_q^j)\notin E(G)$, we define the rectangle 
\[A_{p,q}^{i,j} = [2ir+p+1,2ir+r+p-1]\times [2jr+q+1,2jr+r+q-1].\]
Then we set 
\[\mathcal{R}_A =  \{A_{p,q}^{i,j} ~:~ (v_p^i,v_q^j)\notin E \mbox{ and } i \neq j\}.\]
%
%\[
%\begin{array}{rl}
%    \mathcal{R}_A = & \{[2ir+p+1,2ir+r+p-1]\times[2jr+q+1,2jr+r+q-1]:(v_p^i,v_q^j)\notin E\}
%\end{array}
%\]
%
Note that for every non-edge $(v_p^i,v_q^j)$ of $G$ with $i \neq j$, both $A_{p,q}^{i,j}$ and $A_{q,p}^{j,i}$ are rectangles in $\mathcal{R}_A$.
%Notice that by our definition, $[2jr+q+1,2jr+r+q-1]\times [2ir+p+1,2ir+r+p-1]$ is also in $\mathcal{R}_A$. \todo{phrase better. Something about ordered pairs}

Finally we describe $\mathcal{R}_E$, the set of rectangles that ensures {\bf e}quality of the vertices selected by the solution lines in the strips
$\zeta_v^{2i-2}$, $\zeta_v^{2i-1}$, $\zeta_h^{2i-2}$ and $\zeta_h^{2i-1}$.
%
%. This set of rectangles is to ensure that for all $i\in [k]$, the lines in strips $[2ir+1,2ir+r]\times \mathbb{R},[2ir+r+1,2ir+2r]\times \mathbb{R},\mathbb{R}\times [2ir+1,2ir+r]$ and $\mathbb{R}\times [2ir+r+1,2ir+2r]$ represent the same vertex (if exactly one line is picked in each strip).
%
%Let 
%$\mathcal{R}_E^{x,y} = \{[xr+1,xr+a]\times [yr+r-b,yr+r]:1\leq a,b\leq r-1, a+b = r-1, a,b\in \mathbb{N}\}\cup \{[xr+r-b,xr+r]\times [yr+1,yr+a]:1\leq a,b\leq r-1, a+b = r-1, a,b\in \mathbb{N}\}$. 
%
%
For every pair $x, y \in [0, 2k-1]_\mathbb{Z}$ such that $\lfloor x/2 \rfloor = \lfloor y/2 \rfloor$ and every integer $a \in [2, r]_\mathbb{Z}$ we define the ``top left'' rectangle
\[TL^{x,y}_a = [2r + xr + 1, 2r + xr + a - 1] \times [2r + yr + a, 2r + yr + r]\]
and the ``bottom right'' rectangle
\[BR^{x,y}_a = [2r + xr + a, 2r + xr + r] \times [2r + yr + 1, 2r + yr + a - 1].\]
%\[\mathcal{R}_E^{x,y} = \{[xr+1,xr+a]\times [yr+r-b,yr+r]:1\leq a,b\leq r-1, a+b = r-1, a,b\in \mathbb{N}\}\]
%\[\{[xr+r-b,xr+r]\times [yr+1,yr+a]:1\leq a,b\leq r-1, a+b = r-1, a,b\in \mathbb{N}\}\]
%
We set
%\[{\cal R}_E = \bigcup_{\substack{x,y \in [0, 2k-1]_\mathbb{Z} \mbox{ s.t.} \\ \lfloor x/2 \rfloor = \lfloor y/2 \rfloor \\ a \in  [2, r]_\mathbb{Z}}} \{TL^{x,y}_a, BR^{x,y}_a\}\]
\[{\cal R}_E = \{TL^{x,y}_a, BR^{x,y}_a ~:~ x,y \in [0, 2k-1]_\mathbb{Z} \mbox{ s.t. } \lfloor x/2 \rfloor = \lfloor y/2 \rfloor, ~a \in  [2, r]_\mathbb{Z}\}.\]
%We set
%\[
%\begin{array}{rl}
%    \mathcal{R}_E = & \mathcal{R}_E^{2i,2i}\cup \mathcal{R}_E^{2i+1,2i}\cup \mathcal{R}_E^{2i,2i+1}\cup \mathcal{R}_E^{2i+1,2i+1}
%\end{array}
%\]
This concludes the construction of the instance of {\sc Rectangle Stabbing}, see Figure \ref{fig:hardness-reduction-construction}.
% A visualization of the construction is shown in Figure \ref{fig:hardness-reduction-construction}.

\begin{figure}[H]
    \centering
    \includegraphics[width=0.9\linewidth]{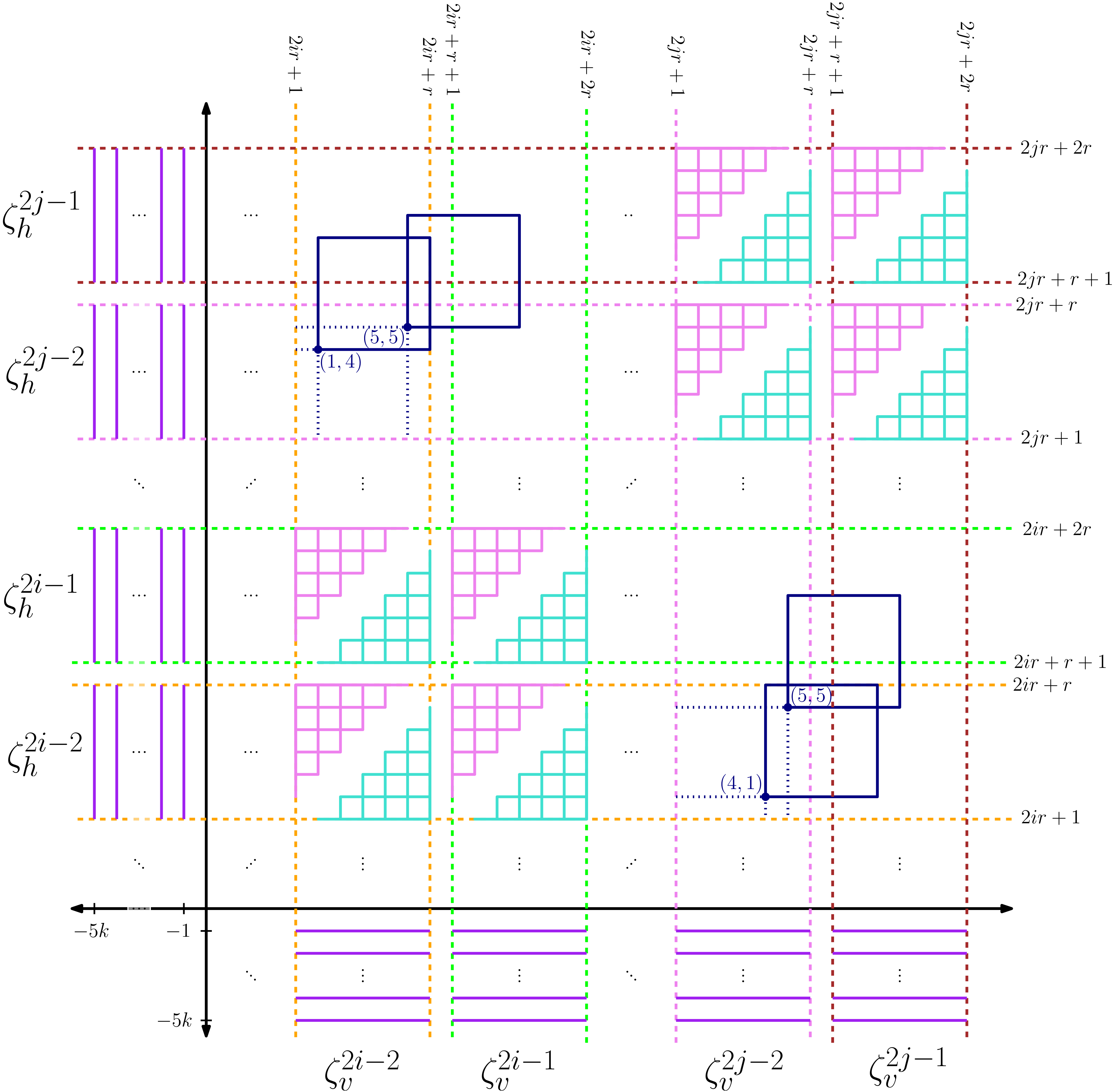}
    \caption{\textit{Example snippet of the reduction construction from {\sc Multicolored Clique} to {\sc Rectangle Stabbing}.} The dark purple rectangles (line segments) with one negative coordinate are in ${\cal R}_F$. The dark blue rectangles intersecting $\zeta_v^{2i-2} \cap \zeta_h^{2j-2}$ and $\zeta_v^{2j-2} \cap \zeta_h^{2i-2}$ are in ${\cal R}_A$. The pink and cyan rectangles arranged like staircases are in ${\cal R}_E$.\\[5pt]
    Notice \textbf{in this example}, $r = 7$ and the bottom left endpoints of the rectangles intersecting $\zeta_v^{2i-2} \cap \zeta_h^{2j-2}$ have an offset of $(1, 4)$ and $(5, 5)$ from $(2ir+1, 2jr+1)$. These represent the non-edges $(v^i_1, v^j_4),\ (v^i_5, v^j_5) \not\in E(G)$ for this example graph $G$.}
    \label{fig:hardness-reduction-construction}
\end{figure}

We make two remarks about the construction in Theorem~\ref{thm-hardness}. First, note that the set ${\cal L}$ of lines produced by the construction in Section~\ref{sec:construction} is simply the set of all lines with integer coordinates. For a finite set of explicitly given lines one could restrict ${\cal L}$ to all vertical lines with integer position in the vertical strips
$\{\zeta_v^x ~:~ x \in [0,2k-1]_{\mathbb{Z}} \}$ and all horizontal lines with integer position in the horizontal strips $\{\zeta_h^x ~:~ x \in [0,2k-1]_{\mathbb{Z}} \}$.

Second, the construction in Section~\ref{sec:construction} produces some degenerate rectangles (that is, rectangles on the form $[a, b] \times [c, d]$ where $a \leq b$, $c \leq d$ and either $a=b$ or $c=d$ or both equalities hold). 
This is simply for ease of notation, and degenerate rectangles are not necessary for the reduction to work. 
Indeed it is easy to see that replacing every rectangle $[a, b] \times [c, d]$ in ${\cal R}$
with the rectangle $[2a, 2b+1] \times [2c, 2d+1]$ to produce the set ${\cal R}'$, and keeping ${\cal L}$ as the set of all lines with integer positions gives an equivalent instance in the following sense. 
A subset ${\cal H}' \cup {\cal V}'$ of lines stabs ${\cal R}'$ if and only if the set $\{l_h^{\lfloor y/2 \rfloor} ~:~ l_h^y \in {\cal H}'\} \cup \{l_v^{\lfloor x/2 \rfloor} ~:~ l_v^x \in {\cal V}'\}$ stabs ${\cal R}$.
Thus, algorithm ${\cal A}$ in the proof of Theorem~\ref{thm-hardness} can be run on ${\cal R}'$ instead of ${\cal R}$, and the statement of Theorem~\ref{thm-hardness} holds even for algorithms that require non-degenerate input rectangles. 

\subsection{Proof of Correctness for the Hardness Reduction}

\begin{lemma}\label{lem:forwardDir}
If there exists a multicolored clique of size $k$ in $G$, then there exists a set of horizontal lines $\mathcal{H}$ and a set of vertical lines $\mathcal{V}$ such that 
$\mathcal{H}\cup\mathcal{V} \subseteq {\cal L}$, $|\mathcal{H}\cup\mathcal{V}| = 4k$ and $\mathcal{H}\cup \mathcal{V}$ stabs $\mathcal{R}$.
\end{lemma}

\begin{proof}
Suppose that $C$ is a multicolored clique of size $k$ in $G$. For each $i \in [k]$ let $r_i$ be such that $\{v_{r_i}^i\} = C \cap V_i$.
%. Since $|C|=k$ and $C$ is multicolored, it has to be the case where $|C\cap V_i|=1$ for each $i\in [k]$. Let $C\cap V_i = \{v_{p_i}^i\}$. 
%
%
%For every $x \in [0, 2k-1]_{\mathbb{Z}}$ 
%let $\hat{h}^x$ be the line $\hat{h}^x = \ell_h^{2r+xr+r_{\lfloor x/2 \rfloor}}$ and 
%$\hat{v}^x$ be the line $\hat{v}^x = \ell_v^{2r+xr+r_{\lfloor x/2 \rfloor}}$. 
%\todo[inline]{maybe redefine the above using $i$ from $1..k$ to avoid the floor in exponent DONE BELOW}
%
For every $i \in [k]_{\mathbb{Z}}$ let
$\hat{h}^{2i-2}$ be the line $\hat{h}^{2i-2} = \ell_h^{2r+(2i-2)r+r_{i}}$,
$\hat{h}^{2i-1}$ be the line $\hat{h}^{2i-1} = \ell_h^{2r+(2i-1)r+r_{i}}$,
$\hat{v}^{2i-2}$ be the line $\hat{v}^{2i-2} = \ell_v^{2r+(2i-2)r+r_{i}}$, and 
$\hat{v}^{2i-1}$ be the line $\hat{v}^{2i-1} = \ell_v^{2r+(2i-1)r+r_{i}}$.
%
%\todo[inline]{maybe redefine the above using $i$ from $1..k$ to avoid the floor in exponent}
%
Observe that for every $x \in [0, 2k-1]_{\mathbb{Z}}$ we have 
$\hat{h}^x$ lies in the strip $\zeta_h^x$ while $\hat{v}^x$ is in the strip $\zeta_v^x$.
We set $\mathcal{H} = \{\hat{h}^x ~:~ x \in [0, 2k-1]_{\mathbb{Z}} \}$ and
$\mathcal{V} = \{\hat{v}^x ~:~ x \in [0, 2k-1]_{\mathbb{Z}} \}$.
%
%
%$ to be the set of lines with coordinates in $\{2ir+p_i:i\in [k]\}\cup \{2ir+r+p_i:i\in [k]\}$, then set $\mathcal{L}$ to be the set of lines with coordinates in $\{2ir+p_i:i\in [k]\}\cup \{2ir+r+p_i:i\in [k]\}$. 
%
Thus $|\mathcal{H}\cup \mathcal{V}| = 4k$. We now prove that all rectangles in ${\cal R}$ are stabbed by $\mathcal{H}\cup \mathcal{V}$.

First we verify that $\mathcal{R}_F$ is stabbed by $\mathcal{H}\cup \mathcal{V}$. 
For every $x \in [0,2k-1]_\mathbb{Z}$ and every $q \in [-5k, -1]_{\mathbb{Z}}$ 
$F_v^{x,q}$ is stabbed by the vertical line $\hat{v}^x$
and $F_h^{x,q}$ is stabbed by the horizontal line $\hat{h}^x$.

%Firstly, let's check the rectangles in $\mathcal{R}_F$. For every rectangle of the form $[2r+xr+1,2r+xr+r]\times [-2x-2,-2x-1],x\in [0,2k-1]\cap \mathbb{N}$, it is stabbed by the vertical line with coordinate $2r+xr+p_{\lceil (x+1)/2\rceil}$. For every rectangle of the form $[-2x-2,-2x-1]\times [2r+xr+1,2r+xr+r] ,x\in [0,2k-1]\cap \mathbb{N}$, it is stabbed by the horizontal line with coordinate $2r+xr+p_{\lceil (x+1)/2\rceil}$.

Next we check that $\mathcal{R}_A$ is stabbed by $\mathcal{H}\cup \mathcal{V}$. Consider a rectangle $A^{i,j}_{p,q} \in \mathcal{R}_A$. Note that $i,j \in [k]$ and $i \neq j$. 
Since $C$ is a clique we have that $r_i\neq p$ or $r_j\neq q$, and therefore
$r_i \leq p-1$ or $r_i \geq p+1$ or $r_j\leq q-1$
%\todo{this should be $\leq$ not $\neq$? if you agree then please change it :) }
or $r_j\geq q+1$. We check the four possible cases:
\begin{itemize}\setlength\itemsep{-.7pt}
\item If $r_i \geq p+1$ then $\hat{v}^{2i-2}$ stabs $A^{i,j}_{p,q}$,
\item if $r_i \leq p-1$ then $\hat{v}^{2i-1}$ stabs $A^{i,j}_{p,q}$,
\item if $r_j \geq q+1$ then $\hat{h}^{2j-2}$ stabs $A^{i,j}_{p,q}$, and
\item if $r_j \leq q-1$ then $\hat{h}^{2j-1}$ stabs $A^{i,j}_{p,q}$.
\end{itemize}
%If $p_i\neq p$, then the vertical line with coordinate $2ir+p_i$ stabs this rectangle when $p_i\geq p+1$, or the vertical line with coordinate $2ir+r+p_i$ stabs this rectangle when $p_i\leq p-1$. If $p_j\neq j$, then the horizontal line with coordinate $2jr+p_j$ stabs this rectangle when $p_j\geq p+1$, or the horizontal line with coordinate $2jr+r+p_j$ stabs this rectangle when $p_j\leq p-1$.

%Secondly, let's check the rectangles in $\mathcal{R}_A$. Fix a rectangle of the form $[2ir+p+1,2ir+r+p-1]\times[2jr+q+1,2jr+r+q-1],(v_p^i,v_q^j)\notin E$. Since $C$ is a clique, either $p_i\neq p$ or $p_j\neq j$. If $p_i\neq p$, then the vertical line with coordinate $2ir+p_i$ stabs this rectangle when $p_i\geq p+1$, or the vertical line with coordinate $2ir+r+p_i$ stabs this rectangle when $p_i\leq p-1$. If $p_j\neq j$, then the horizontal line with coordinate $2jr+p_j$ stabs this rectangle when $p_j\geq p+1$, or the horizontal line with coordinate $2jr+r+p_j$ stabs this rectangle when $p_j\leq p-1$.
Finally we verify that  $\mathcal{H}\cup \mathcal{V}$ stabs $\mathcal{R}_E$.
Let $x$ and $y \in [0, 2k-1]_\mathbb{Z}$ be such that $\lfloor x/2 \rfloor = \lfloor y/2 \rfloor$, and $a$ be an element of $[2, r]_\mathbb{Z}$. 
Let $i = \lfloor x/2 \rfloor = \lfloor y/2 \rfloor$.
We now verify that the rectangles $TL_a^{x,y}$ and $BR_a^{x,y}$ are stabbed by $\mathcal{H}\cup \mathcal{V}$.
\begin{itemize}\setlength\itemsep{-.7pt}
\item If $r_i \geq a$ then $TL^{x,y}_a$ is stabbed by $\hat{h}^y$, and 
\item if $r_i \leq a-1$ then $TL^{x,y}_a$ is stabbed by $\hat{v}^x$.
\item If $r_i \geq a$ then $BR^{x,y}_a$ is stabbed by $\hat{v}^x$, and
\item if $r_i \leq a-1$ then $BR^{x,y}_a$ is stabbed by $\hat{h}^y$.
\end{itemize}
Thus all rectangles in ${\cal R}$ are stabbed by $\mathcal{H}\cup \mathcal{V}$, concluding the proof.
\end{proof}

\begin{lemma}\label{lem:reverseDir}
%There exists an algorithm that takes as input $G, k, V_1, \ldots, V_k$, ${\cal R}$, ${\cal L}$, a real $\epsilon > 0$ and a set $\mathcal{H}$ of horizontal lines and a set $\mathcal{V}$ of vertical lines such that $|\mathcal{H}\cup\mathcal{V}| \leq 5k-\epsilon k$ and $\mathcal{H} \cup \mathcal{V}$ stabs $\mathcal{R}$, runs in polynomial time, and outputs a multicolored clique $C$ of size at least $\epsilon k$ in $G$.
There exists a polynomial time algorithm that takes as input $G, k, V_1, \ldots, V_k$, ${\cal R}$, ${\cal L}$, a real $\epsilon > 0$, a set $\mathcal{H}$ of horizontal lines, and a set $\mathcal{V}$ of vertical lines such that $|\mathcal{H} \cup \mathcal{V}| \leq 5k-\epsilon k$ and $\mathcal{H} \cup \mathcal{V}$ stabs $\mathcal{R}$, and outputs a multicolored clique $C$ of size at least $\epsilon k$ in $G$.
\end{lemma}

\begin{proof}
Let $\mathcal{H}$ be a set of horizontal lines and $\mathcal{V}$ be a set of vertical lines such that $|\mathcal{H} \cup \mathcal{V}| \leq 5k-\epsilon k$ and $\mathcal{H} \cup \mathcal{V}$ stabs $\mathcal{R}$.

\begin{claim}\label{clm:atLeastFour}
For every $x \in [0, 2k-1]_\mathbb{Z}$, 
$\mathcal{H}$ contains at least one line in the strip $\zeta^x_h$ and 
$\mathcal{V}$ contains at least one line in the strip $\zeta^x_v$.
\end{claim}
\begin{proof}
Suppose for contradiction that there exists an  $x \in [0, 2k-1]_\mathbb{Z}$ such that $\mathcal{H}$ does not contain any lines in the strip $\zeta^x_h$. For every $q \in [-5k, -1]_\mathbb{Z}$ the rectangle $F_h^{x,q}$ is not stabbed by ${\cal H}$, since every horizontal line stabbing $F_h^{x,q}$ is in $\zeta^x_h$. Thus the set $\{F_h^{x,q} ~:~ q \in [-5k, -1]_\mathbb{Z}\}$ is stabbed by ${\cal V}$. But then $|{\cal V}| \geq 5k$, because no vertical line stabs two rectangles in  $\{F_h^{x,q} ~:~ q \in [-5k, -1]_\mathbb{Z}\}$. This contradicts that $|\mathcal{H}\cup\mathcal{V}| \leq 5k-\epsilon k$. This shows that for every $x \in [0, 2k-1]_\mathbb{Z}$, $\mathcal{H}$ contains at least one line in the strip $\zeta^x_h$. The proof that $\mathcal{V}$ contains at least one line in the strip $\zeta^x_v$ is symmetric.  
\end{proof}

Let $A$ be the subset of all integers $i \in [k]$ such that each strip $\zeta^{2i-2}_h$, $\zeta^{2i-1}_h$, $\zeta^{2i-2}_v$ and $\zeta^{2i-1}_v$ contains precisely one line from $\mathcal{H} \cup \mathcal{V}$.
We claim that $|A| \geq \epsilon \cdot k$. Aiming towards a contradiction, suppose $|A| < \epsilon \cdot k$. Then, by Claim~\ref{clm:atLeastFour} for each $i \in [k]$ we have that the strips 
$\{\zeta^{2i-2}_h, \zeta^{2i-1}_h, \zeta^{2i-2}_v, \zeta^{2i-1}_v\}$ contain at least $4$ lines from $\mathcal{H}\cup\mathcal{V}$ if $i \in A$ and at least $5$ lines from $\mathcal{H}\cup\mathcal{V}$ if $i \notin A$.
It follows that $|\mathcal{H}\cup\mathcal{V}| \geq 4|A| + 5(k - |A|) = 5k - |A| > 5k - \epsilon k$, contradicting $|\mathcal{H}\cup\mathcal{V}| \leq 5k-\epsilon k$. Hence  $|A| \geq \epsilon \cdot k$.

For every $i \in A$ we define the integers $v\langle i,-\rangle$, $v\langle i,+\rangle$, $h\langle i,-\rangle$, $h\langle i,+\rangle$ as the unique integers in $[r]$ such that 
\begin{itemize}\setlength\itemsep{-.7pt}
    \item $\ell_v^{2ir+v\langle i,-\rangle} \in {\cal V}$ and lies in $\zeta^{2i-2}_v$,

    \item $\ell_v^{2ir+r+v\langle i,+\rangle} \in {\cal V}$ and lies in $\zeta^{2i-1}_v$,

    \item $\ell_h^{2ir+h\langle i,-\rangle} \in {\cal H}$ and lies in $\zeta^{2i-2}_h$,

    \item $\ell_h^{2ir+r+h\langle i,+\rangle} \in {\cal H}$ and lies in $\zeta^{2i-1}_h$.
\end{itemize}

\begin{claim}\label{clm:equalPos}
For every $i \in A$ we have $v\langle i,-\rangle = h\langle i,+\rangle$.
\end{claim}

\begin{proof}
Suppose for contradiction that $v\langle i,-\rangle > h\langle i,+\rangle$.
Then $v\langle i,-\rangle \geq 2$.
Since $\lfloor (2i-2)/2 \rfloor = \lfloor (2i-1)/2 \rfloor$, we have that
$TL^{2i-2,2i-1}_{v\langle i,-\rangle}$ is a rectangle in ${\cal R}_E$.
Since $TL^{2i-2,2i-1}_{v\langle i,-\rangle} \subseteq \zeta^{2i-2}_v \cap \zeta^{2i-1}_h$,
$\ell_v^{2ir+v\langle i,-\rangle}$ is the unique line in ${\cal V}$ which lies in $\zeta^{2i-2}_v$,
and
$\ell_h^{2ir+r+h\langle i,+\rangle}$ is the unique line in ${\cal H}$ which lies in $\zeta^{2i-1}_h$, it follows that 
$TL^{2i-2,2i-1}_{v\langle i,-\rangle}$ is not stabbed by $\mathcal{H}\cup\mathcal{V} - \{\ell_v^{2ir+v\langle i,-\rangle}, \ell_h^{2ir+r+h\langle i,+\rangle} \}$.
However, the largest $x$-coordinate of a point in $TL^{2i-2,2i-1}_{v\langle i,-\rangle}$ is 
\[2r + (2i-2)r + v\langle i,-\rangle - 1 < 2ir +  v\langle i,-\rangle\mbox{,}\]
and therefore $\ell_v^{2ir+v\langle i,-\rangle}$ does not stab  $TL^{2i-2,2i-1}_{v\langle i,-\rangle}$.
Furthermore, the smallest $y$-coordinate of a point in $TL^{2i-2,2i-1}_{v\langle i,-\rangle}$ is
\[2r + (2i-1)r + v\langle i,-\rangle = 2ir+r+ v\langle i,-\rangle > 2ir+r+h\langle i,+\rangle\mbox{,}\]
and therefore $\ell_h^{2ir+r+h\langle i,+\rangle} \in {\cal H}$ does not stab $TL^{2i-2,2i-1}_{v\langle i,-\rangle}$. This contradicts that 
$TL^{2i-2,2i-1}_{v\langle i,-\rangle-1}$ is stabbed by $\mathcal{H}\cup\mathcal{V}$, so we conclude that $v\langle i,-\rangle \leq h\langle i,+\rangle$.

Next suppose for contradiction that $v\langle i,-\rangle < h\langle i,+\rangle$. 
Then $h\langle i,+\rangle \geq 2$.
Since $\lfloor (2i-2)/2 \rfloor = \lfloor (2i-1)/2 \rfloor$, we have that
$BR^{2i-2,2i-1}_{h\langle i,+\rangle}$ is a rectangle in ${\cal R}_E$.
Again $BR^{2i-2,2i-1}_{h\langle i,+\rangle}$ is not stabbed by $\mathcal{H}\cup\mathcal{V} - \{\ell_v^{2ir+v\langle i,-\rangle}, \ell_h^{2ir+r+h\langle i,+\rangle} \}$.
However the largest $y$-coordinate of $BR^{2i-2,2i-1}_{h\langle i,+\rangle}$ is
\[2r + (2i-1)r + h\langle i,+\rangle - 1 < 2ir+r+h\langle i,+\rangle\mbox{,}\]
and so $\ell_h^{2ir+r+h\langle i,+\rangle}$ does not stab $BR^{2i-2,2i-1}_{h\langle i,+\rangle}$
Furthermore the smallest $x$-coordinate of $BR^{2i-2,2i-1}_{h\langle i,+\rangle}$ is
\[2r + (2i-2)r + h\langle i,+\rangle = 2ir+h\langle i,+\rangle  > 2ir+v\langle i,-\rangle\mbox{,}\]
and so $\ell_v^{2ir+v\langle i,-\rangle}$ does not stab $BR^{2i-2,2i-1}_{h\langle i,+\rangle}$. This contradicts that $BR^{2i-2,2i-1}_{h\langle i,+\rangle}$ is stabbed by $v\langle i,-\rangle \leq h\langle i,+\rangle$, so we conclude that $v\langle i,-\rangle \geq h\langle i,+\rangle$.  
We have now shown that $v\langle i,-\rangle \leq h\langle i,+\rangle$ and that $v\langle i,-\rangle \geq h\langle i,+\rangle$, which implies that $v\langle i,-\rangle = h\langle i,+\rangle$, concluding the proof of the claim.
\end{proof}

Proofs symmetrical to that of Claim~\ref{clm:equalPos} show $v\langle i,-\rangle = h\langle i,-\rangle$, $v\langle i,+\rangle = h\langle i,-\rangle$, and $v\langle i,+\rangle = h\langle i,+\rangle$ as well. Hence we obtain that for every $i \in A$ we have 
\[v\langle i,-\rangle = v\langle i,+\rangle = h\langle i,-\rangle = h\langle i,+\rangle\mbox{.}\]

For each $i \in A$ we define $r\langle i\rangle = v\langle i,-\rangle$ and $C = \{v^i_{r\langle i\rangle} ~:~ i \in A\}$. Clearly $|C| = |A| \geq \epsilon k$ and so it remains to show that $C$ is in fact a clique in $G$. Suppose for contradiction that there exist distinct $i, j \in A$ such that $(v^i_{r\langle i\rangle}, v^j_{r\langle j\rangle})$ is not an edge of $G$.
%
%Suppose it is not, then there is some $i,j\in A$ such that $(v_{p_i}^i,v_{p_j}^j)\notin E$. 
 %
Consider the rectangle $A^{i,j}_{r\langle i\rangle, r\langle j\rangle}$. By construction it is not stabbed by any line with integer coordinates which is not in one of the strips $\zeta_v^{2i-2}$, $\zeta_v^{2i-1}$, $\zeta_h^{2j-2}$, or $\zeta_h^{2j-2}$.
Thus, no lines in  $\mathcal{H}\cup\mathcal{V} - \{\ell^{2ir + r\langle i\rangle}_v, \ell^{2ir + r + r\langle i\rangle}_v, \ell^{2jr + r\langle j\rangle}_h, \ell^{2jr + r + r\langle j\rangle}_h \}$ stab $A^{i,j}_{r\langle i\rangle, r\langle j\rangle}$.
However the rectangle $A^{i,j}_{r\langle i\rangle, r\langle j\rangle}$ is defined precisely so it is contained in the interior of the rectangle defined by the four lines $\{\ell^{2ir + r\langle i\rangle}_v, \ell^{2ir + r + r\langle i\rangle}_v, \ell^{2jr + r\langle j\rangle}_h, \ell^{2jr + r + r\langle j\rangle}_h \}$. Hence $A^{i,j}_{r\langle i\rangle, r\langle j\rangle}$ is not stabbed by any of $\{\ell^{2ir + r\langle i\rangle}_v, \ell^{2ir + r + r\langle i\rangle}_v, \ell^{2jr + r\langle j\rangle}_h, \ell^{2jr + r + r\langle j\rangle}_h \}$, and therefore it is not stabbed by  $\mathcal{H}\cup\mathcal{V}$ yielding the desired contradiction. 

Both the construction of $A$ from the input and the construction of $C$ from $A$ clearly take polynomial time. This concludes the proof. 
\end{proof}

We are now ready to prove Theorem~\ref{thm-hardness}. 
%\todo[inline]{if you want to re-sstate the theorem before proving it do it here!}

\begin{proof}[Proof of Theorem~\ref{thm-hardness}]
Suppose that there exists an $\epsilon > 0$, function $f : \mathbb{N} \rightarrow \mathbb{N}$, 
and an algorithm ${\cal A}$ which takes as input an instance $({\cal R}, {\cal L})$, an integer $k$,  runs in time $f(k) \cdot n^{O(1)}$, where $n = |\mathcal{R}| + |\mathcal{L}|$, and either concludes that no subset of at most $k$ lines of ${\cal L}$ stabs ${\cal R}$ or finds a set ${\cal L}^* \subseteq {\cal L}$ of size at most $\left(\frac{5}{4}-\epsilon\right)k$ that stabs ${\cal R}$.

We describe an approximation algorithm for {\sc Multicolored Clique}. Given an instance $G$,  $k$ and partition $V_1, V_2, \ldots V_k$ of $V(G)$ , the algorithm constructs the instance $({\cal R}, {\cal L}, 4k)$ of {\sc Rectangle Stabbing} using the reduction described in Section~\ref{sec:construction}.
It then runs the algorithm ${\cal A}$ on this instance. If ${\cal A}$ returns that ${\cal R}$ cannot be stabbed with at most $4k$ lines from ${\cal L}$, the algorithm returns that $G$ does not contain a multicolored clique of size $k$. This is correct by Lemma~\ref{lem:forwardDir}.
If ${\cal A}$ returns a set of at most $\left(\frac{5}{4}-\epsilon\right)k$ lines that stab ${\cal R}$, the algorithm applies Lemma~\ref{lem:reverseDir} to produce a multicolored clique $C$ in $G$ of size at least $\epsilon k$.
The running time of the algorithm is upper bounded by $f(4k)n^{O(1)}$.
Hence, by applying Theorem~\ref{prop:mccHardness} with $h(k) = 1/\epsilon$ to the resulting algorithm for {\sc Multicolored Clique} we obtain that FPT $=$ W[1].
\end{proof}

\section{Conclusion and future work} \label{sec-conclusion}

We have shown that {\sc Rectangle Stabbing} admits a $\frac{7}{4}$-approximation algorithm with a runtime of $k^{O(k)} \cdot n^{O(1)}$, where $n = |\mathcal{R}| + |\mathcal{L}|$. The approximation ratio of the algorithm improves upon the best-known for polynomial-time algorithms, and the running time is faster than what can be achieved for exact parameterized algorithms, assuming $\text{FPT} \neq \text{W[1]}$. Additionally, we showed that, assuming $\text{FPT} \neq \text{W[1]}$, no algorithm running in $f(k) \cdot n^{O(1)}$ time, for any computable function $f$, can achieve an approximation ratio better than $\frac{5}{4} - \varepsilon$, for any $\varepsilon > 0$.

In light of our work, a natural open problem is to find the tractability boundary of parameterized approximation for {\sc Rectangle Stabbing}, in particular find the constant $c$ such that there exists a parameterized $c$-approximation, while a parameterized $c'$-approximation for every $c'<c$ would imply FPT $=$ W[1] or violate the (Gap) Exponential Time Hypothesis. Another interesting direction is to determine whether the approximation ratio of $2$ is optimal when restricted to polynomial-time algorithms.

\bibliography{bibliography}

% \appendix
% \section{Appendix}
% \input{appendix}

\end{document}